\documentclass[12pt]{article}
\usepackage{graphicx, amsmath,amsfonts,amsthm, euscript,amscd}
\usepackage{braket}
\usepackage{tikz}
\usetikzlibrary{decorations.pathreplacing}
\usepackage{dsfont}
\usepackage{subfigure}
\usepackage[shortlabels]{enumitem}
\usepackage{tikz,}

\graphicspath{ {./images/} }
\usepackage{verbatim}
\usepackage{color}
\usepackage{ulem}
\newtheorem{theorem}{Theorem}[section]
\newtheorem{proposition}[theorem]{Proposition}
\newtheorem{lemma}[theorem]{Lemma}
\newtheorem{corollary}{Corollary}[theorem]
\newtheorem{definition}{Definition}[section]
\newtheorem{remark}{Remark}[theorem]
\newtheorem{example}{Example}[section]

\usepackage{graphicx, amsmath}
\usepackage{color}
\definecolor{Blue}{rgb}{0.3,0.3,0.9}

\begin{document}

\normalsize

\title{Time-inhomogeneous Quantum Markov Chains with Decoherence on Finite State Spaces}

\maketitle

\vskip 6pt
\author{ Chia-Han Chou and  Wei-Shih Yang}
\vskip 3pt
\indent Department of Mathematics, Temple University,\\
\indent     Philadelphia, PA 19122

\vskip 6pt
Email: chia-han.chou@temple.edu, yang@temple.edu

\vskip 12pt

KEY WORDS:  Quantum Markov Chain; Quantum Decoherence; Time-Inhomogeneous

\vskip12pt
AMS classification Primary: 81Q10, Secondary: 82B10
%81Q10(Self-adjoint operator theory in quantum theory),82B10 (Quantum equilibrium statistical mechanics)

PACS numbers: 03.65.Yz, 05.30.-d, 03.67.Lx, 02.50.Ga
%03.65.Yz (Decoherence; open systems; quantum statistical methods), 05.30.-d (Quantum statistical mechanics), 03.67.Lx (Quantum computation architectures and implementations), 02.50.Ga (Markov Processes)
\vskip 12pt

\begin{abstract}
 %{\color{red}  [delete] In quantum computation theory, quantum Markov chains and quantum walks have been utilized by many quantum search algorithms which provide improved performance over their classical counterparts. More recently, due to the importance of the quantum decoherence phenomenon, decoherent quantum walks and their applications have been studied on a wide variety of structures. In this paper, [end delete] } 
  We introduce and study time-inhomogeneous quantum Markov chains with parameter $\zeta \ge 0$ and decoherence parameter $0 \leq p \leq 1$ on finite spaces and their large scale equilibrium properties. Here $\zeta$ resembles the inverse temperature in the annealing random process and $p$ is the decoherence strength of the quantum system.  Numerical evaluations show that if $ \zeta$ is small, then quantum Markov chain is ergodic for all $0 < p \le 1$ and if $ \zeta $ is large, then it has  multiple limiting distributions for all $0 < p \le 1$. In this paper, we prove the ergodic property in the high temperature region $0 \le \zeta \le 1$. We expect that the phase transition occurs at the critical point $\zeta_c=1$. For coherence case $p=0$, a critical behavior of periodicity also appears at critical point  $\zeta_o=2$.

  %that convergence of decoherent time-inhomogeneous quantum Markov chain on finite state spaces. 

\end{abstract}

\section{Introduction}
\setcounter{equation}{0}
In order to develop more efficient algorithms for tackling a wide variety of problems in classical computer science, researchers started utilizing randomness techniques such as Ulam and von Nuemann’s Markov Chain Monte Carlo (MCMC) method \cite{MetroUlam} in 1940s. This method was later refined and made well known as the Metropolis-Hastings algorithm \cite{Hastings} with applications in different areas. Even though Monte Carlos methods could sometimes return incorrect solutions with given probability, the key idea behind the methods was that the true solution can be approximated with high probability by repeating Monte Carlo simulations.

More recently, the notion of quantum computation has gained popularity,
"qubit" takes a complex unit instead of "bit" the usual binary values of zero and one. To preserve a cohesive quantum system, the family of qubits comprising the memory of the computer go through unitary evolution, rather than the traditional system of gates in classical computation theory. The state of the quantum system can be observed, and collapsing the system to one unique state from a superposition of various state after the completion of each algorithm. The probability of observing any given state after observation is proportional to the absolute value squared of the amplitude of the system at that state. So, a false solution may in fact be observed which is similar to Monte-Carlo methods. However, if the algorithm is cleverly constructed, the correct solution is observed with significant likelihood.

Due to the quantum mechanical nature of quantum computation, new types of quantum  algorithms have appeared. Moreover, these algorithms are more efficient than existing classical algorithms because the run times are better. For instance, both integer factorization and discrete logarithms
undergo an exponential speedup using Shor’s algorithm \cite{Shor}. Not only an exponential improvement, Grover’s search algorithm provides a quadratic
speedup over any known classical search algorithm \cite{Grover}, and on a discrete space, Grover’s algorithm is defined by discrete-time quantum walk, which is the natural extension of a Markov chain driven classical walk to the quantum setting. 

On the other hand, if a quantum system were perfectly isolated, it would maintain coherence indefinitely, but it would be impossible to manipulate or investigate it. If it is not perfectly isolated, for example during a measurement, coherence is shared with the environment and appears to be lost with time which is called quantum decoherence. This concept was first introduced by H. Zeh \cite{ZehDecoh} in 1970, and then formulated mathematically for quantum walks by T. Brun \cite{BrunDecoh}. For both coin and position space decoherent Hadamard walk, K. Zhang proved in \cite{ZhangQuantumLimit} that with symmetric initial conditions, it has Gaussian limiting distribution. More recently, the fact that the limiting distribution of the rescaled position discrete-time quantum random walks with general unitary operators subject to only coin space decoherence is a convex combination of normal distributions under certain conditions is proved by S. Fan, Z. Feng, S. Xiong and W. Yang \cite{YangConvexNormal}. The decoherent quantum analogues of Markov chains and random walks on a finite 
%{\color{red} [delete] and discrete infinite space [end delete] } 
space will be defined and elaborated in this paper.

In fact, classical Markov chain limit theorems for the discrete time walks are well known and have had important applications in related areas \cite{Rurrett} and \cite{Mixing}. However, the primary goal of this paper is to examine the limiting behavior of the new model, discrete time-inhomogeneous quantum walk with decoherence on finite spaces 
%{\color{red} [delete] and infinite  spaces,} 
and generalize the results from the classical theorems to the quantum analogues. In this paper, we introduce and  study time-inhomogeneous quantum Markov chains with parameter $\zeta \ge 0$ decoherence parameter $0 \le p \le 1$ on finite spaces and their large scale equilibrium properties. Here $\zeta$ resembles the inverse temperature in the annealing random process and $p$ is the decoherence strength of the quantum system.  Numerical evaluations show that if $ \zeta$ is small, then quantum Markov chain is ergodic for all $0 < p \le 1$ and if $ \zeta $ is large, then it has  multiple limiting distributions for all $0 < p \le 1$. In this paper, we prove the ergodic property in the high temperature region $0 \le \zeta \le 1$. We conjecture that the  phase transition is at the critical point $\zeta_c=1$.  For coherence case $p=0$, a critical behavior of periodicity also appears at critical point  $\zeta_o=2$.

For the time homogeneous case $\zeta =0$ and $o < p \le 1$, the limiting distribution has been obtained by Lagro, Yang and Xiong \cite{LagroQuantumLimit} under very general conditions similar to the Perron-Frobenius type of conditions. Our paper extends \cite{LagroQuantumLimit} to time-inhomogeneous case, but with stronger assumptions on the transition matrices.

This paper is organized as follows. In Section \ref{sec:notations and definitions}, we set the notations, definitions and introduce  the model.  In Section \ref{sec:path integral},  we develop  a path integral formula for time-inhomogeneous decoherent quantum Markov chains, Theorem  \ref{th:Compond Markov Chain}. In Section \ref{sec:convergence to equilibrium}, we obtain a limiting theorem for classical time-inhomogeneous Markov chains, Theorem \ref{ClassicConver}. In Section \ref{sec:convergence of density operators},  we use these theorems to prove our  main result, Theorem \ref{finalEqui}. In Section \ref{sec:simulations}, we provide numerical evaluation to support our conjectures. In Section \ref{sec:conclusion}, we  make our conclusions and discuss  some problems for further study.

%{\color{red} [delete] we obtain first-return properties of the time-inhomogeneous Markov chain with decoherence are derived rigorously in detail, and prove the convergence to equilibrium of the decoherent quantum Markov chains with time-inhomogeneous unitary operators in general finite spaces using path integral formulas. Moreover, the equilibrium theory is illustrated by numerical simulations.  [end delete]}

\section{Notations and Definitions}\label{sec:notations and definitions}
In classical probability, a random walk on $\mathbb{Z}$ is a Markov process described by a
stochastic transition matrix $T$. On the other hand, for a quantum walk, instead of the transition matrix, the evolution of the system is described by a unitary operator U acting on a Hilbert space H. Several different models for quantum walks have been popularized.
The two most commonly used are the coined walk of Aharonov et al \cite{AharonovQuanGraph} and the quantum markov chain of Szegedy \cite{SzegedyMarkovC}. Recently, S. Fan, Z. Feng, S. Xiong and W. Yang et al. \cite{YangConvexNormal} demonstrated that under certain conditions, the limiting distribution of the rescaled position discrete-time decoherent quantum coined walks is a convex combination of normal distributions. All quantum walks elaborated here in this section are based on homogeneous coined Markov chains.

We consider the quantum states $|1\rangle$, ..., $|m\rangle$. If $m=2$, then it represents the head and tail respectively when we flip a coin. Let $H$ be the Hilbert space spanned by the orthonormal basis, $|1\rangle$, ..., $|m\rangle$, 
$$ H=span\{|1\rangle, ..., |m\rangle\}.$$
 %Let's consider the states $|1\rangle$ and $|2\rangle$ which represent the head and tail respectively when you flip a coin. Then, we suppose that they are orthonormal, and the coined Hilbert space $H$ is defined by 
%$$ H=span\{|1\rangle,|2\rangle\},$$
%which means that any element in $H$ can be written as linear combination of $|1\rangle$ and $|2\rangle$. 
We denote the inner product of the space by $\langle|i\rangle,|j\rangle\rangle:=\langle i|j\rangle$, and $|i\rangle^*=\langle i|$ for $x=1,2, ..., m$. 
%Therefore, we have for instance,
%$$\langle 1|1\rangle=1 \text{, and } \langle 1|2\rangle=0$$
%On the other hand, 
For n=1, 2, 3, ..., let $U_n:H \rightarrow H$ be a unitary operator acting on the Hilbert space $H$ itself. 
%Recall that a unitary operator $U$ satisfies the property that $UU^*=U^*U=I$ (where $I$ denotes the identity operator, and $U^*$ is the adjoint of $U$).
%
%For n=1, 2, 3, ..., let $U:H \rightarrow H$ be a unitary operator acting on the Hilbert space $H$ itself. Recall that a unitary operator $U$ satisfies the property that $UU^*=U^*U=I$ (where $I$ denotes the identity operator, and $U^*$ is the adjoint of $U$).
Now, in order to define a decoherent quantum Markov chain over the Hilbert space $H$,
we consider the decoherence parameter $p \in [0,1]$, and define for $i=1, ..., m$
$$ A_i=\sqrt{p} \cdot |i\rangle\langle i|,$$
and 
$$ A_0=\sqrt{1-p}\cdot I$$
Note that $\{A_0,A_1, ..., A_m\}$ has the property that $\sum_i {A_i}^*A_i=I$, and is called a measurement over the space $H$. This notion will be generally defined later.

\begin{definition}\label{InhomoDecohOperator}
	 Let $L(H)$ be the set of linear operators on $H$. 
	%$|x\rangle\in H$, 
 For $\rho \in L(H)$,
	%$\rho:=|x\rangle\langle x|$, 
	we define the $n$-th step time-inhomogeneous decoherent operator $\Phi_n$ such as 
	$$ \Phi_n(\rho)=\sum_{i=0}^{m}A_{i}U_n\rho U_{n}^{*}A_{i}^{*}$$
\end{definition}

Now, suppose that $x\in \{1, ..., m\}$ is the initial position, and $\rho_0=|x\rangle \langle x|$, and the $n$ step time-inhomogeneous quantum Markov chain is defined as
\begin{eqnarray}
\rho_{n}=\underset{n \text{ times}}{\Phi_n\cdots\Phi_1}(\rho_0),
\end{eqnarray}
and, the probability of getting $y\in \{1, ..., m\}$ from the initial position $x\in \{1, ..., m\}$ after $n$ steps is defined as the trace of $|y\rangle\langle y|\rho_n$, we denote it by
\begin{eqnarray}
P_n(x,y):=Tr\Big(|y\rangle\langle y|\rho_n\Big).\label{InhomoDecohProbability}
\end{eqnarray}

\begin{example}\label{HadamardMC}
	Let $U_n=U$ such as,
	\begin{eqnarray}
	U=
	\left[
	\begin{array}{ c c  }
	\frac{1}{\sqrt 2} &\frac{1}{\sqrt 2} \\
	\frac{1}{\sqrt 2}&-\frac{1}{\sqrt 2}
	\end{array} \right]
	\end{eqnarray}
	Then the quantum Markov chain is called the  time-homogeneous Hadamard (fair coined) quantum Markov chain.
\end{example}

%%In this paper, instead of homogeneous unitary operators $U$, let's first consider 
\begin{example}\label{time-inho HadamardMC}
Let $\displaystyle U_n:H \rightarrow H$, be unitary defined by 
$$ U_n=\begin{bmatrix} 
\sqrt{1-\frac{\lambda}{n^{\zeta}}} & \sqrt{\frac{\lambda}{n^{\zeta}}} \\
\sqrt{\frac{\lambda}{n^{\zeta}}} & -\sqrt{1-\frac{\lambda}{n^{\zeta}}} 
\end{bmatrix},$$
 time inhomogeneous unitary operators, where $\lambda$ and $\zeta$ are non negative real numbers. 
%%And we define the decoherent time-inhomogeneous quantum Markov chain with 
%%\begin{eqnarray}
%%\Phi_n(\rho)=\sum_{i=0}^{2}A_{i}U_n\rho {U_n}^{*}A_{i}^{*}
%%\end{eqnarray}
%%with $ A_0=\sqrt{1-p}\cdot I$, and  
%%$ A_i=\sqrt{p} \cdot |i\rangle\langle i|,$ for $p \in [0,1]$, and $i=1,2$. Similarly
%%\begin{eqnarray}
%%\rho_{n}=\Phi_n\cdots\Phi_1(\rho_0),
%%\end{eqnarray}
%%Therefore, the probability of getting $y\in \{1,2\}$ from  $x\in \{1,2\}$ at time $n$  is defined as
%%\begin{eqnarray}
%%P_n(x,y):=Tr\Big(|y\rangle\langle y|\rho_n\Big).
%%\end{eqnarray}
Note that we can easily obtain the homogeneous case from this setting by letting $\zeta=0$. For example, if $\zeta=0$ and $\lambda=\frac{1}{2}$, we have the fair coined quantum Markov chain from Example \ref{HadamardMC}.
\end{example}

\begin{example}\label{time-inho m by m HadamardMC}

Let 
$H=span\{|1\rangle,...,|m\rangle\}$ with $m\in\mathbb{N}$, and $U_n=e^{i\frac{G}{\sqrt{n^\zeta}}}$ where $G$ is an $m\times m$ self-adjoint  matrix, where $\zeta$ is a  non-negative real number.

\end{example}

\section{Compound Markov chain representation and basic properties}\label{sec:path integral}
Let's start by proving some basic properties before we go to the equilibrium convergence theorem. First, if we suppose that the quantum Markov chain we defined is completely decoherent, $p=1$, we have $$\Phi_{n}(\rho)=\sum_{i=1}^{m}A_{i}U_{n}\rho U_{n}^{*}A_{i}^{*}$$ and $$P_n(i,j)=Tr\Big(|j\rangle\langle j|\Phi_{n}|i\rangle\langle i|\Big)$$ where $\rho$ is any $m$ by $m$ density matrix, and $i,j = 1, ..., m$. Using the fact that ${|k\rangle, k=1,2,...}$ are orthonormal basis, we obtain 

$$P_n(i,j)=\sum_{k=1}^{m}Tr\Big(|j\rangle\langle j|k\rangle\langle k|U_n|i\rangle\langle i|U_n^*|k\rangle\langle k|\Big)$$	$$=\sum_{k=1}^{m}Tr\Big(\langle j|k\rangle\langle k|U_n|i\rangle\langle i|U_n^*|k\rangle\langle k|j\rangle)\Big)$$
$$=Tr\Big(\langle j|U_n|i\rangle\langle i|U_n^*|j\rangle\Big)=U_n(j,i)U_n^*(i,j)=|U_n(j,i)|^2$$

If $\zeta=0$ and $p=1$, then it reduces to time-homogeneous Markov chain and it is well known that $\rho_n\rightarrow\rho_{\infty}$ when $n\rightarrow \infty$ by the Ergodic Theorem for finite-state Markov chains. Our focus now will be for the case of $\zeta>0$ or $0<p<1$.

By Kolomogorov $0$-$1$ law, for $0<p\leq 1$, there exists an infinite sequence of measurements when $t\rightarrow\infty$. Let $X_1,X_2,...$ be  the outcomes of the measurements. In the following Proposition \ref{markovian}, we will show that given the measurement times, $\{X_n\}_{n=1}^{\infty}$ is a time-inhomogeneous Markov chain.

\begin{definition}\label{MarkovQ}
	Let $\{T_n\}_{n=1}^{\infty}$ be i.i.d. geometric random variables with mean $\frac{1}{p}$.  Let $\sigma_0=0$ and $\sigma_n=T_1+\cdots+T_n$.  For $n \ge 1$, we define a Markov transition matrix
	\begin{eqnarray}
	Q_{\sigma_{n-1}}(i,j):=\big|\langle j| U_{\sigma_{n-1}+T_n}\cdots U_{\sigma_{n-1}+1} |i\rangle\big|^2
	\end{eqnarray}
\end{definition}	

%Note that we can write $\Phi_{n}=(1-p)\Phi_{n_0}+p\Phi_{n_1}$, where $\Phi_{n_0}=U_n\rho U_n^*$  is the unitary part of the quantum operation,and $\displaystyle\Phi_{n_1}=\sum_{i=1}^{2}|i\rangle \langle i|U_{n}\rho U_{n}^{*}|i\rangle \langle i|$ is the markovian part of the quantum operation. And we have a more general version of the previous proposition.
%\begin{prop}
%	Suppose that $\displaystyle\prod_{k=1}^{\infty}c(\Phi_{n_1})=0$ with $\displaystyle c(P^*)=\max_{x,y}\frac{1}{2}\sum_{z}|P(z,x)-P(z,y)|$, show that $\Phi_{n}\Phi_{n_1}(\rho)\rightarrow\rho_{\infty}$ for all $0<p\leq1$. (This proposition is not useful for now)
%\end{prop}

\begin{proposition}\label{markovian}
	Let $\{T_n\}_{n=1}^{\infty}$ be i.i.d. geometric random variables with mean $\frac{1}{p}$, and let $\sigma_0=0$ and $\sigma_n=T_1+\cdots+T_n$. Let $X_1, X_2, ...$ be the outcomes of the measurements at measurement times $\sigma_0,  \sigma_1, ... $. If  $\rho_{0}=|i\rangle\langle i|$,  then 
	\begin{enumerate}
		\item[(a)] $\displaystyle P_{\sigma}\big(X_1=i_1,X_2=i_2,...,X_n=i_n\big)=\big|\langle i_n| U_{\sigma_{n}}\cdots U_{\sigma_{n-1}+1} |i_{n-1}\rangle\big|^2\cdots$
		$$\cdots\big|\langle i_1| U_{\sigma_{1}}\cdots U_{1} |i\rangle\big|^2$$
		\item[(b)] $\displaystyle P_{\sigma}\big(X_n=i_n |X_1=i_1,X_2=i_2,...,X_{n-1}=i_{n-1}\big)=P_{\sigma}\big(X_n=i_n|X_{n-1}=i_{n-1}\big)$
		$$=Q_{\sigma_{n-1}}\big(i_{n-1},i_n\big)$$
	\end{enumerate}	
Here $i_0=i$ and $P_{\sigma}$ denotes the conditional probability given $\{\sigma_1, \sigma_2, ...\}$.
\end{proposition}
\noindent
\begin{proof}[Proof of Proposition \ref{markovian}]
Let $\displaystyle\Phi_{k}(\rho)=\sum_{i=0}^{m}A_iU_k\rho U_k^*A_i$, $\rho\in \mathcal{E}$, where $\mathcal{E}=$the set of all density operators on $H$. The quantum operation of the partially decoherent quantum process at step $k$ is given by $\Phi_{k}$.

Suppose that the initial state is $\rho_0=|i\rangle\langle i|$. Then the state at time $t$ is $$\rho_t=\Phi_{t}\cdots\Phi_{1}(\rho_0)$$
The probability at time $t$, the system is found at state $|j\rangle$ is given by
$$Tr(|j\rangle\langle j|\rho_t)=\langle j|\sum_{j_t=0}^{m}\cdots\sum_{j_1=0}^{m}A_{j_t}U_t\cdots A_{j_2}U_2A_{j_1}U_1|i\rangle \langle i|U_1^*A_{j_1}^*\cdots U_t^*A_{j_t}^*|j\rangle$$
$$=\sum_{j_t=0}^{m}\cdots\sum_{j_1=0}^{m}\langle j|A_{j_t}U_t\cdots A_{j_2}U_2A_{j_1}U_1|i\rangle \langle i|U_1^*A_{j_1}^*\cdots U_t^*A_{j_t}^*|j\rangle$$
$$=\sum_{j_t,...,j_1=0}^{m}\big|\langle j|A_{j_t}U_t\cdots A_{j_2}U_2A_{j_1}U_1|i\rangle\big|^2$$
In each term of above sum, let $0<\sigma_1<\sigma_2<...<\sigma_n\leq t$ be exactly the times that $j_{\sigma_k}=1$, $2$, ...,  or $m$, and $j_s=0$ for all $s\neq \sigma_1,...,\sigma_{n}$. $0<\sigma_1<\sigma_2<...<\sigma_n$ are called the decoherence times, and we put $i_k=j_{\sigma_{k}}$. So, the sum can be written as,
$$Tr(|j\rangle\langle j|\rho_t)=\sum_{n=0}^{\infty}\text{ }\sum_{0=\sigma_{0}<...<\sigma_{n}\leq t}p^{n}q^{t-n}\sum_{i_n=1}^{m}\cdots\sum_{i_1=1}^{m}\big|\langle j|U_t\cdots U_{\sigma_{n}+1}|i_n\rangle\big|^2\cdot$$
\begin{eqnarray}\label{pathIntegralOri}
\cdot\big|\langle i_n|U_{\sigma_{n}}\cdots U_{\sigma_{n-1}+1}|i_{n-1}\rangle\big|^2\cdots\big|\langle i_1|U_{\sigma_{1}}\cdots U_{1}|i\rangle\big|^2
\end{eqnarray}

Now, we prove Part (a) of the proposition.
$$P\big(X_1=i_1,...,X_n=i_n\big)=\sum_{t=1}^{\infty}P\big(X_1=i_1,...,X_n=i_n,\sigma_n=t\big)$$
%$$=\sum_{t=1}^{\infty}\text{ }\sum_{0<\sigma_1<\cdots<\sigma_n= t}P_{\sigma_{1},...,\sigma_{n},t}\big(i_1,...,i_n,j\big) \text{ where } t=\sigma_n \text{ and } j=i_n$$
By Equation (\ref{pathIntegralOri}), we have 
$$=\sum_{t=1}^{\infty}\text{ }\sum_{0<\sigma_1<\cdots<\sigma_n= t}p^{n}q^{t-n}\big|\langle i_n|U_{\sigma_{n}}\cdots U_{\sigma_{n-1}+1}|i_{n-1}\rangle\big|^2\cdots\big|\langle i_1|U_{\sigma_{1}}\cdots U_{1}|i\rangle\big|^2$$
%$$\sum_{0=\sigma_{0}<...<\sigma_{n}\leq t}p^{n}q^{t-n}\sum_{i_n=1}^{m}\cdots\sum_{i_1=1}^{m}\big|\langle j|U_t\cdots U_{\sigma_{n}+1}|i_n\rangle\big|^2\cdot$$
%$$\cdot\big|\langle i_n|U_{\sigma_{n}}\cdots U_{\sigma_{n-1}+1}|i_{n-1}\rangle\big|^2\cdots\big|\langle i_1|U_{\sigma_{1}}\cdots U_{1}|i\rangle\big|^2$$
$$=\sum_{0<\sigma_1<\cdots<\sigma_n < \infty}pq^{\sigma_1-\sigma_{0}-1}pq^{\sigma_2-\sigma_{1}-1}...pq^{\sigma_n-\sigma_{n-1}-1}\big|\langle i_n|U_{\sigma_{n}}\cdots U_{\sigma_{n-1}+1}|i_{n-1}\rangle\big|^2\cdots\big|\langle i_1|U_{\sigma_{1}}\cdots U_{1}|i\rangle\big|^2.$$

On the other hand, 
$$P\big(X_1=i_1,...,X_n=i_n\big)=E\Big[P_{\sigma}(X_1=i_1, ..., X_n=i_n)\Big],$$
and $$P(T_1=\sigma_1, T_1+T_2=\sigma_2, ..., T_1+...+T_n=\sigma_n)=pq^{\sigma_1-\sigma_{0}-1}pq^{\sigma_2-\sigma_{1}-1}...pq^{\sigma_n-\sigma_{n-1}-1}.$$
Therefore, we have 
$$P_{\sigma}(X_1=i_1, ..., X_n=i_n)=\big|\langle i_n|U_{\sigma_{n}}\cdots U_{\sigma_{n-1}+1}|i_{n-1}\rangle\big|^2\cdots\big|\langle i_1|U_{\sigma_{1}}\cdots U_{1}|i\rangle\big|^2$$

To prove (b), by  (a) and definition, we have 
$$P_{\sigma}\big(X_1=i_1,X_2=i_2,...,X_n=i_n\big)=Q_{\sigma_0} (i,i_1)\cdots Q_{\sigma_{n-1}}(i_{n-1},i_n)$$
Then,
$$P_{\sigma}\big(X_n=i_n|X_1=i_1,...,X_{n-1}=i_{n-1}\big)=\frac{P_{\sigma}\big(X_1=i_1,...,X_{n-1}=i_{n-1},X_n=i_n\big)}{P_{\sigma} \big(X_1=i_1,...,X_{n-1}=i_{n-1}\big)}$$
$$=Q_{\sigma_{n-1}}(i_{n-1},i_n)$$
We also note that
$$P_{\sigma }\big(X_n=i_n|X_{n-1}=i_{n-1}\big)=\frac{P_{\sigma}\big(X_n=i_n,X_{n-1}=i_{n-1}\big)}{P_{\sigma}\big(X_{n-1}=i_{n-1}\big)}$$
$$=\frac{\sum_{i_1,...,i_{n-2}=1}^{m}Q_{\sigma_0}(i,i_1)\cdots Q_{\sigma_{n-2}}(i_{n-2},i_{n-1})Q_{\sigma_{n-1}}(i_{n-1},i_n)}{\sum_{i_1,...,i_{n-2}=1}^{m}Q_{\sigma_0}(i,i_1)\cdots Q_{\sigma_{n-2}}(i_{n-2},i_{n-1})}=Q_{\sigma_{n-1}}(i_{n-1},i_n) $$
\end{proof}

\begin{remark} 
\begin{enumerate}
	\item[(a)] It follows from Proposition \ref{markovian} (b) that given measurement times $\sigma$, $X_1, X_2,...$ is a time in-homogeneous Markov chain with transition probability
	$$P_{\sigma}\big(X_n=j|X_{n-1}=i\big)=Q_{\sigma_{n-1}}(i,j),\ 1\leq i,j\leq m,$$
	and this is illustrated in the following diagram
	
	\begin{tikzpicture}[%
	every node/.style={
		font=\scriptsize,
		% Better alignment, see https://tex.stackexchange.com/questions/315075
		text height=1.5ex,
		text depth=0.25ex,
	},
	]\label{timediagram}
	% draw horizontal line   
	\draw[->] (0,0) -- (13,0);
	
	% draw vertical lines
	\foreach \x in {0,1,2,4,8,10,11,12}{\draw (\x cm,3pt) -- (\x cm,-3pt);}
	
	% place axis labels
	\node[anchor=north] at (0,0) {$0$};
	\node[anchor=north] at (1,0) {$1$};
	\node[anchor=north] at (1.5,0) {...};
	\node[anchor=north] at (2,0) {$\sigma_1$};
	\node[anchor=north] at (2,-.4) {$\uparrow$};
	\node[anchor=north] at (2,-1) {\large$X_1$};
	\node[anchor=north] at (3,0) {...};
	\node[anchor=north] at (4,0) {$\sigma_2$};
	\node[anchor=north] at (4,-.4) {$\uparrow$};
	\node[anchor=north] at (4,-1) {\large$X_2$};
	\node[anchor=north] at (6,0) {...};
	
	\node[anchor=north] at (8,0) {$\sigma_{n-1}$};
	\node[anchor=north] at (8,-.4) {$\uparrow$};
	\node[anchor=north] at (8,-1) {\large$X_{n-1}$};
	\node[anchor=north] at (10,0) {$\sigma_n$};
	\node[anchor=north] at (10,-.4) {$\uparrow$};
	\node[anchor=north] at (10,-1) {\large$X_{n}$};
	\node[anchor=north] at (9,0) {...};
	\node[anchor=north] at (11,0) {$t$};
	
	\node[anchor=north] at (12,0) {$\sigma_{n+1}$};
	\node[anchor=north] at (12,-.4) {$\uparrow$};
	\node[anchor=north] at (12,-1) {\large$X_{n+1}$};
	\node[anchor=north] at (13,0) {time};

	% draw curly braces and add their labels
	\draw[decorate,decoration={brace,amplitude=5pt}] (0,.5) -- (2,.5)
	node[anchor=south,midway,above=4pt] {$T_1$};
	
	\draw[decorate,decoration={brace,amplitude=5pt}] (2,.5) -- (4,.5)
	node[anchor=south,midway,above=4pt] {$T_2$};
	
	\draw[decorate,decoration={brace,amplitude=5pt}] (8,.5) -- (10,.5)
	node[anchor=south,midway,above=4pt] {$T_n$};
	
	\draw[decorate,decoration={brace,amplitude=5pt}] (10,.5) -- (12,.5)
	node[anchor=south,midway,above=4pt] {$T_{n+1}$};
	\end{tikzpicture}
	
	\item[(b)] This proposition gives a probabilistic interpretation of decoherence time $0<\sigma_{0}<\sigma_{1}<...<\sigma_{n}<...$ as a sequence of arrival times of independence Bernoulli trials with success probability $p$, and $\{X_n\}_{n=1}^\infty$ can be viewed as discrete version of a compound Poisson process.
	
	\end{enumerate}

\end{remark}

	For $\sigma_{n}\leq t$, and from Equation (\ref{pathIntegralOri}), we obtain an expression for the probability that the system is found at state $j$ with exact decoherent time $\sigma_{1},...\sigma_{n}$ and outcomes $i_1, i_2,...,i_n$. We will formualte this as our main representation theorem as follows.

\begin{definition}\label{pathIntegral}

Let $t \ge 0$ and $n_t=\max\{n \ge 0; \sigma _n \le t \} $ be the number of occurences  before or at time $t$. Let 
$$W_{\sigma _{n_t}}(i, j)=\big|\langle j|U_t\cdots U_{\sigma_{n_t-1}+1}|i\rangle\big|^2$$

%
%	The probability that the system is found at state $j$ with exact decoherent time $\sigma_{1},...,\sigma_{n}$ and outcomes $i_1, i_2,...,i_n$, in this case $X_1=i_1,...X_n=i_n$ if $\sigma_{n}\leq t$, for the decoherent quantum random walk for $0<p\leq 1$ can be written as the path integral formula:
%	$$P_{\sigma_{1},\sigma_{2},...,\sigma_{n},t}\big(i_1,i_2,...,i_n,j\big):=p^nq^{t-n}\big|\langle j|U_t\cdots U_{\sigma_{n}+1}|i_n\rangle\big|^2\cdots
%	\big|\langle i_1|U_{\sigma_{1}}\cdots U_{1}|i\rangle\big|^2$$	
\end{definition}

From Equation (\ref{pathIntegralOri})  and Definition \ref{pathIntegral}, we have the following
\begin{theorem}[Compound Markov Chain Representation]\label{th:Compond Markov Chain}
Suppose $\rho_{0}=|i\rangle\langle i|$. Then

$$Tr(|j\rangle\langle j|\rho_t)=E[Q_{\sigma_0}Q_{\sigma_1}...Q_{\sigma_{n_t-1}}W_{\sigma_{n_t}}(i,j)]$$

%$$Tr(|j\rangle\langle j|\rho_t)=\sum_{n=0}^{\infty}\text{ }\sum_{0=\sigma_{0}<...<\sigma_{n}\leq t}p^{n}q^{t-n}\sum_{i_n=1}^{m}\cdots\sum_{i_1=1}^{m}\big|\langle j|U_t\cdots U_{\sigma_{n}+1}|i_n\rangle\big|^2\cdot$$
%\begin{eqnarray}\label{representation}
%\cdot\big|\langle i_n|U_{\sigma_{n}}\cdots U_{\sigma_{n-1}+1}|i_{n-1}\rangle\big|^2\cdots\big|\langle i_1|U_{\sigma_{1}}\cdots U_{1}|i\rangle\big|^2
%\end{eqnarray}

\end{theorem}

Note that in the above theorem, the matrices $Q$  and $W$ are pure quantum random walks. In between decohenernces times, the successive processes are coherent quantum walks. 
%The path integral formula defined in Definition \ref{pathIntegral} 
The Compound Markov Chain Representation Theorem \ref{th:Compond Markov Chain} not only gives a very intuitive expression  for the $0 < p \le 1$  decoherent quantum walks with the probability of a specific path from the initial state to the target state, but also provides a very useful tool for generating numerical simulations, as we shall see in Section \ref{sec:simulations}.

\begin{proof}[Proof of Theorem \ref{th:Compond Markov Chain}] The right hand side of the theorem, 

 $$ E[Q_{\sigma_0}Q_{\sigma_1}...Q_{\sigma_{n_t-1}}W_{\sigma_{n_t}}(i,j)] $$
$$=
\sum_{n=0}^{\infty}\text{ }\sum_{0=\sigma_{0}<...<\sigma_{n}\leq t < \sigma_{n+1}}P(T_1=\sigma_1, T_1+T_2=\sigma_2, ..., T_1+...+T_n =\sigma_n, T_1+...+T_{n+1}=\sigma_{n+1}) $$
$$\sum_{i_n=1}^{m}\cdots\sum_{i_1=1}^{m}\big|\langle j|U_t\cdots U_{\sigma_{n}+1}|i_n\rangle\big|^2
\cdot\big|\langle i_n|U_{\sigma_{n}}\cdots U_{\sigma_{n-1}+1}|i_{n-1}\rangle\big|^2\cdots\big|\langle i_1|U_{\sigma_{1}}\cdots U_{1}|i\rangle\big|^2 $$
$$=
\sum_{n=0}^{\infty}\text{ }\sum_{0=\sigma_{0}<...<\sigma_{n}\leq t < \sigma_{n+1}}p^{n}q^{t-n}q^{\sigma_{n+1}-t-1}p\sum_{i_n=1}^{m}\cdots\sum_{i_1=1}^{m}\big|\langle j|U_t\cdots U_{\sigma_{n}+1}|i_n\rangle\big|^2
\cdot$$
$$\big|\langle i_n|U_{\sigma_{n}}\cdots U_{\sigma_{n-1}+1}|i_{n-1}\rangle\big|^2\cdots
\big|\langle i_1|U_{\sigma_{1}}\cdots U_{1}|i\rangle\big|^2 $$
$$=
\sum_{n=0}^{\infty}\text{ }\sum_{0=\sigma_{0}<...<\sigma_{n}\leq t}p^{n}q^{t-n}\sum_{i_n=1}^{m}\cdots\sum_{i_1=1}^{m}\big|\langle j|U_t\cdots U_{\sigma_{n}+1}|i_n\rangle\big|^2
\cdot$$
$$\big|\langle i_n|U_{\sigma_{n}}\cdots U_{\sigma_{n-1}+1}|i_{n-1}\rangle\big|^2\cdots
\big|\langle i_1|U_{\sigma_{1}}\cdots U_{1}|i\rangle\big|^2 $$
$$= Tr(|j\rangle\langle j|\rho_t),$$
by Equation (\ref{pathIntegralOri}).

\end{proof}

\section{Convergence to equilibrium}\label{sec:convergence to equilibrium}

We first show a limiting theorem for time-inhomogeneous Markov chains in classical probability as follows.

%If we look at the proof of this theorem, we can generalize the technique to this general theorem in classical probability.  
\begin{theorem}\label{ClassicConver}
	Let $\Pi_{ij}=\Pi_{1j} \ge 0$ for all $i, j$ and $\sum_{j}\Pi_{ij}=1$. Let $P_k$ be the Markov transition matrices on a finite states space $\Sigma=\{1, 2, ..., m\}$. Suppose there is a $k_0$ such that for all $k \ge k_0$, $\Pi P_k=
	%P_k\Pi=
	\Pi$ and $P_k\geq\delta_k\Pi$, where $0 \le \delta_k \le 1$. If $\displaystyle\prod_{k=k_0}^{\infty}(1-\delta_k)=0$, then $P_s\cdots P_n\rightarrow\Pi$ as $n\rightarrow\infty$, for all $1 \le s \le k_0$.
\end{theorem}

\begin{proof}[Proof of Theorem \ref{ClassicConver}] Note that since $P_k$ are Markov transition matrices, if $P_k\geq\delta_k\Pi$, for some $\delta_k >0$, then $ \delta_k \le 1$. Moreover, if $\delta_k=1$, then $P_k=\Pi$. So if there is a $k \ge k_0$ such that $\delta_k=1$, then $P_1\cdots P_n=\Pi$, since $P_k\Pi=\Pi$ and by assumption $\Pi P_k=\Pi$. Therefore it is sufficient to consider the case $0 \le \delta_k <1$. 

We also note that since $P_k$ are Markov transition matrices, if
 $$P_{k_0}P_{k_0+1}\cdots P_n \rightarrow\Pi,$$ then $$P_s\cdots P_n\rightarrow\Pi,$$ as $n\rightarrow\infty$, for all $1 \le s \le k_0$.

For $n \ge k_0$, we have 
$$P_n=\alpha_n\Big(\frac{P_n-\delta_n\Pi}{\alpha_n}\Big)+\delta_n\Pi=\alpha_n\tilde{P}_n+\delta_n\Pi$$
Since $\alpha_n=1-\delta_n$ and $\tilde{P}_n$ is a Markov transition matrix, and by assumption, we have the following
\begin{eqnarray}
	 \alpha_n+\delta_n=1 
	\end{eqnarray}
	\begin{eqnarray}
	\Pi\cdot\Pi=\Pi
	\end{eqnarray}
	\begin{eqnarray}
	\displaystyle\tilde{P_n}\cdot\Pi=\frac{P_n-\delta_n\Pi}{\alpha_n}\Pi=\frac{P_n\Pi-\delta_n\Pi^2}{\alpha_n}=\frac{\Pi-\delta_n\Pi}{\alpha_n}=\frac{1-\delta_n}{\alpha_n}\Pi=\Pi 
	\end{eqnarray}
	\begin{eqnarray}
	\Pi\cdot\tilde{P_n}=\Pi 
	\end{eqnarray}
%	\begin{eqnarray}
%	 \tilde{Q} is doubly stochastic.
%	\end{eqnarray}

Therefore, we have 
$$P_{k_0}\cdots P_n=\prod_{k=k_0}^{n}\big[\alpha_n\tilde{P}_n+\delta_n\Pi\big]=\prod_{k=k_0}^{n}\alpha_k\tilde{P}_k+\big(1-\prod_{k=k_0}^{n}\alpha_k\big)\Pi$$
The last equality can be proved by induction, if $n=k_0+1$
$$\big[\alpha_{k_0}\tilde{P}_{k_0}+(1-\alpha_{k_0})\Pi\big]\cdot\big[\alpha_{k_0+1}\tilde{P}_{k_0+1}+(1-\alpha_{k_0+1})\Pi\big]$$
$$=\alpha_{k_0}\alpha_{k_0+1}\tilde{P}_{k_0}\tilde{P}_{k_0+1}+(1-\alpha_{k_0})\alpha_{k_0+1}\Pi\tilde{P}_{k_0+1}+\alpha_{k_0}(1-\alpha_{k_0+1})\tilde{P}_{k_0}\Pi+(1-\alpha_{k_0})(1-\alpha_{k_0+1})\Pi^2$$
$$=\alpha_{k_0}\alpha_{k_0+1}\tilde{P}_{k_0}\tilde{P}_{k_0+1}+\Pi-\alpha_{k_0}\alpha_{k_0+1}\Pi+\alpha_{k_0}\Pi-\alpha_{k_0}\alpha_{k_0+1}\Pi-\Pi+\alpha_{k_0}\alpha_{k_0+1}\Pi+\Pi-\alpha_{k_0}\Pi=$$
$$\alpha_{k_0}\alpha_{k_0+1}\tilde{P}_{k_0}\tilde{P}_{k_0+1}+(1-\alpha_{k_0}\alpha_{k_0+1})\Pi$$
And now, let's assume that the formula is true for $n$, and prove it for $n+1$,
$$\prod_{k=k_0}^{n+1}\big[\alpha_n\tilde{P}_n+\delta_n\Pi\big]=\Big[\prod_{k=k_0}^{n}\alpha_k\tilde{P}_k+\big(1-\prod_{k=k_0}^{n}\alpha_k\big)\Pi\Big]\Big[\alpha_{n+1}\tilde{P}_{n+1}+(1-\alpha_{n+1})\Pi\Big]$$
$$=\prod_{k=k_0}^{n+1}\alpha_k\tilde{P}_k+\prod_{k=k_0}^{n}\alpha_k(1-\alpha_{n+1})\Pi+\big(1-\prod_{k=k_0}^{n}\alpha_k\big)\alpha_{n+1}\Pi+\big(1-\prod_{k=k_0}^{n}\alpha_k\big)(1-\alpha_{n+1})\Pi$$
$$=\prod_{k=k_0}^{n+1}\alpha_k\tilde{P}_k+\prod_{k=k_0}^{n}\alpha_k\Pi-\prod_{k=1}^{n+1}\alpha_k\Pi+\alpha_{n+1}\Pi-\prod_{k=k_0}^{n+1}\alpha_k\Pi+\Pi-\alpha_{n+1}\Pi-\prod_{k=k_0}^{n}\alpha_k\Pi+\prod_{k=1}^{n+1}\alpha_k\Pi$$
$$=\prod_{k=k_0}^{n+1}\alpha_k\tilde{P}_k+\big(1-\prod_{k=k_0}^{n+1}\alpha_k\big)\Pi$$

For an $m \times m$ matrix $B$, we denote by  $|| B||_{\infty}$ the operator norm from $(\mathbb{C}^m, ||\cdot||_{\infty})$ with sup norm to itself. Note that if $P$ is a Markov transition matrix, then $||P||_{\infty} =1$.

%For a vector $v=(v_1, ..., v_m)$, we denote by $||v||_{\infty}=\sup _i |v_i|$ the sup norm of $v$. For an $m \times m$ matrix $B$, we denote by  $$|| B||_{\infty}=\sup _{ v \ne 0} \frac{||Bv||_{\infty} }{||v||_{\infty}}$$ the sup norm of $B$. 
%Note that 
%for any $m \times m$ matrix $B=(b_{ij})$, 
%\begin{eqnarray}
%||B||_{\infty} \le \sup_i\sum_j|b_{ij}| \label{B infinite norm}
%\end{eqnarray}
%and  for any $m \times m$ matrices $A $ and $B$, we have 
%\begin{eqnarray}
%||AB||_{\infty} \le ||A||_{\infty} \||B||_{\infty} \label{AB infinite norm}
%\end{eqnarray}
%Note that if $P$ is a Markov transition matrix, then $||P||_{\infty} = \sup_i\sum_jP_{ij}=1$.

Then we have 
$$
||\prod_{k=k_0}^{n}\alpha_n\tilde{P}_n||_{\infty} \le \prod_{k=k_0}^{n}\alpha_n
$$
We observe that if $\displaystyle \prod_{k=k_0}^{\infty}\alpha_k=0$, then $\displaystyle P_{k_0}P_{k_0+1}\cdots P_n\rightarrow\Pi$, and we have obtained  the conclusion.
\end{proof}

We will need the following estimates for a unitary semigroup.
\begin{lemma}\label{le:lower bound semigroup}
Let $G$ be an $m\times m$ self-adjoint  matrix, and $U=e^{i\theta G}$ where  $\theta$ is a  non-negative real number. Suppose there exists a constant $\epsilon_0 >0$ such that $|G_{jk}|  \ge \epsilon_0$, for all $j,k$. If $\theta \le \frac{ \epsilon_0}{4||G||_{\infty}^2} \wedge \frac{ 1}{4||G||_{\infty}}$, then 

(a) for all $j\ne k$, we have 

$$ \frac{1}{2} \theta \epsilon_0 \le |U_{jk}| \le      2 \theta  || G ||_{\infty},$$

(b) for all $j$, we have $$\frac{1}{2} \le |U_{jj}|.$$

\end{lemma}

\begin{proof}[Proof of Lemma \ref{le:lower bound semigroup}] (a) Let $ j \ne k$. For the first inequality, we have  
$$|U_{jk}|=|(e^{i\theta G})_{jk}|=|[e^{i\theta G}-I]_{jk}|=|i\theta G_{jk} +\sum_{n=2}^{\infty} \frac{(i \theta G )^n_{jk}}{n !} |$$
$$ \ge |i\theta G_{jk}| - |\sum_{n=2}^{\infty} \frac{(i \theta G )^n_{jk}}{n !} |$$
\begin{eqnarray}
 \ge \theta \epsilon_0 - \sum_{n=2}^{\infty} \frac{ \theta ^n || G ||_{\infty}^n}{n !} \label{eq:lower bound1}
\end{eqnarray}
The second term of the above is bounded by 
\begin{eqnarray} \sum_{n=2}^{\infty} \frac{ \theta ^n || G ||_{\infty}^n}{n !}  \le \sum_{n=2}^{\infty}  \theta ^n || G ||_{\infty}^n=\frac{\theta ^2 || G ||_{\infty}^2}{1-\theta  || G ||_{\infty}} \le 2 \theta ^2 || G ||_{\infty}^2, \label{eq: lower bound2}
\end{eqnarray}
since $\theta \le \frac{1}{4||G||_{\infty}}$. Therefore, we have 
\begin{eqnarray}
|U_{jk}| \ge \theta \epsilon_0 - 2 \theta ^2 || G ||_{\infty}^2 \ge \frac{1}{2}\theta \epsilon_0,
\end{eqnarray}
since $\theta \le \frac{\epsilon_0}{4||G||_{\infty}^2}$.

For the second inequality in (a), we have 
$$|U_{jk}|=|(e^{i\theta G})_{jk}|=|[e^{i\theta G}-I]_{jk}|=|\sum_{n=1}^{\infty} \frac{(i \theta G )^n_{jk}}{n !} |$$
\begin{eqnarray}
 &\le\sum_{n=1}^{\infty} \frac{ \theta ^n || G ||_{\infty}^n}{n !} \label{eq:upper  bound1}
 &\le \frac{\theta  || G ||_{\infty}}{1-\theta  || G ||_{\infty}} \le 2 \theta  || G ||_{\infty}, \label{eq: upper  bound2}
\end{eqnarray}
%\begin{eqnarray} \sum_{n=2}^{\infty} \frac{ \theta ^n || G ||_{\infty}^n}{n !}  \le \sum_{n=2}^{\infty}  \theta ^n || G ||_{\infty}^n=\frac{\theta ^2 || G ||_{\infty}^2}{1-\theta  || G ||_{\infty}} \le 2 \theta ^2 || G ||_{\infty}^2, \label{eq: upper  bound2}
%\end{eqnarray}
since $\theta \le \frac{1}{4||G||_{\infty}}$.

(b) For any $j$, we have $$|U_{jj}|=|(e^{i\theta G})_{jj}|=|1 +\sum_{n=1}^{\infty} \frac{(i \theta G )^n_{jj}}{n !} |$$
$$ \ge 1- |\sum_{n=1}^{\infty} \frac{(i \theta G )^n_{jk}}{n !} |$$
\begin{eqnarray}
 \ge 1 - \sum_{n=1}^{\infty} \frac{ \theta ^n || G ||_{\infty}^n}{n !} \label{eq:lower bound3}
\end{eqnarray}
The second term of the above is bounded by 
\begin{eqnarray} \sum_{n=1}^{\infty} \frac{ \theta ^n || G ||_{\infty}^n}{n !}  \le \sum_{n=1}^{\infty}  \theta ^n || G ||_{\infty}^n=\frac{\theta  || G ||_{\infty}}{1-\theta  || G ||_{\infty}} \le 2 \theta  || G ||_{\infty}, \label{eq: lower bound4}
\end{eqnarray}
since $\theta \le \frac{1}{4||G||_{\infty}}$. Therefore, we have 
\begin{eqnarray}
|U_{jj}| \ge 1 - 2 \theta  || G ||_{\infty} \ge \frac{1}{2},
\end{eqnarray}
since $\theta \le \frac{1}{4||G||_{\infty}}$.

\end{proof}

Before we prove the convergence, let's first observe that $Q_{\sigma_n}$ defined in Definition \ref{MarkovQ} is doubly stochastic which will be useful later in the proofs. 
\begin{proposition}\label{DoubleStochQ}
	$Q_{\sigma_{n-1}}$ is doubly stochastic for all $n  \ge 1$.
	
\end{proposition}

\begin{proof}[Proof of Proposition \ref{DoubleStochQ}]
Note that a matrix $A=(a_{ij})$, $a_{ij} \ge 0$,  is doubly stochastic if
$$\sum_{i}a_{ij}=1\ \text{ for all } j$$
and
$$\sum_{j}a_{ij}=1\ \text{ for all } i$$
So, we have for all $j=1, 2, ..., m$,
$$\sum_{i=1}^{m}Q_{\sigma_{n-1}}(i,j)=\sum_{i=1}^{m}\langle j|U_{\sigma_{n-1}+T_n}\cdots U_{\sigma_{n-1}+1}|i\rangle\langle i|U_{\sigma_{n-1}+1}^*\cdots U_{\sigma_{n-1}+T_n}^*|j\rangle$$
$$=\sum_{i=1}^{m}Tr\big[|j\rangle\langle j|U_{\sigma_{n-1}+T_n}\cdots U_{\sigma_{n-1}+1}|i\rangle\langle i|U_{\sigma_{n-1}+1}^*\cdots U_{\sigma_{n-1}+T_n}^*\big]$$
$$=Tr\big(|j\rangle\langle j|\big)=1$$
And also, $\displaystyle \sum_{j=1}^{m}Q_{\sigma_{n-1}}(i,j)=1$ for all $i=1,2, ..., m$ by similar argument.
\end{proof}

Finally, we have the equilibrium property
\begin{proposition}\label{equi}
Let 
$H=span\{|1\rangle,...,|m\rangle\}$ with $m\in\mathbb{N}$, and $U_n=e^{i\frac{G}{\sqrt{n^\zeta}}}$ where $G$ is an $m\times m$ self-adjoint  matrix, where  $\zeta$ is a  non-negative real number. Suppose there exists a constant $\epsilon_0 >0$ such that $|G_{ij}|  \ge \epsilon_0$, for all $i, j$. 
	Let $\Pi$ be an $m \times m$ matrix with $\Pi_{ij}=\frac{1}{m}$, for all $i, j $. If $0 <  \zeta \leq 1$ and $0<p\leq 1$, then for all $0 \le s < \infty$, 
	$$ Q_{\sigma_s} Q_{\sigma_{s+1}} Q_{\sigma_{s+2}}\cdots Q_{\sigma_{n-1}} \rightarrow \Pi,$$
%	=\begin{bmatrix} 
%	1/2 & 1/2 \\
%	1/2 & 1/2 
%	\end{bmatrix},$$
	almost surely, as $n\rightarrow\infty$. 
	
\end{proposition}

\begin{proof}[Proof of Proposition \ref{equi} ]

Step 1. We first show that with probability 1, $Q_{\sigma_{n-1}} \ge \delta_{n-1} \Pi$, where $\delta_{n-1}= \frac{m \epsilon_0^2T_n^2}{4(\sigma_{n})^{\zeta}}$, for sufficeintly large $n$.

Let $$\theta_0=\frac{ \epsilon_0}{4||G||_{\infty}^2} \wedge \frac{ 1}{4||G||_{\infty}}$$
Let $$\theta = \frac{1}{(\sigma_{n-1}+1) ^ {\zeta /2}} +...+ \frac{1}{(\sigma_{n-1}+T_n) ^ {\zeta /2}}.$$ Then we have 
\begin{eqnarray}\label{eq:theta bounds}
\frac{T_n}{\sigma_{n}^{\zeta /2}} \le \theta \le \frac{T_n}{\sigma_{n-1}^{\zeta /2}}
\end{eqnarray}
This implies that with probability one, $\theta \to 0$ as $n \to \infty$. Therefore, with probability 1, there is an $n_0 >s$ such that $\theta < \theta_0$, for all $n \ge n_0 >s$.

By definition, 
\begin{eqnarray}
&Q_{\sigma_{n-1}}(i, j) =\langle j|U_{\sigma_{n-1}+T_n}\cdots U_{\sigma_{n-1}+1}|i\rangle\langle i|U_{\sigma_{n-1}+1}^*\cdots U_{\sigma_{n-1}+T_n}^*|j\rangle\\
&=|\langle j|U_{\sigma_{n-1}+T_n}\cdots U_{\sigma_{n-1}+1}|i\rangle|^2=|e^{i \theta G}(i, j)|^2\\
 &\ge \frac{\epsilon _0 ^2 \theta ^2}{4},
\end{eqnarray}
for $i \ne j$, by Lemma \ref{le:lower bound semigroup}, if $\theta \le \theta_0$. For $i=j$, by Lemma \ref{le:lower bound semigroup}, we have a lower bound $1/4$. Combining (\ref{eq:theta bounds}), we have $Q_{\sigma_{n-1}} \ge \delta_{n-1} \Pi$, where $\delta_{n-1}= \frac{m \epsilon_0^2T_n^2}{4(\sigma_{n})^{\zeta}}$, for some $n_0 >s$  with  $n \ge n_0$.

Now
$$\prod_{k=n_0}^{n}\alpha_k=\prod_{k=1}^{n}(1-\delta_k)=\prod_{k=n_0}^{n}\Big[1-\frac{m \epsilon_0^2T_k^2}{4(\sigma_{k})^{\zeta}}\Big]$$
$$=e^{\sum_{k=n_0}^{n}\ln \Big[1-\frac{m \epsilon_0^2T_k^2}{4(\sigma_{k})^{\zeta}}\Big]}\leq e^{-\sum_{k=n_0}^{n}\frac{m \epsilon_0^2T_n^2}{4(\sigma_{n})^{\zeta}}}$$
since $\ln(1-x)\leq-x$ for $0<x<1$.

%Observe that each $T_k\geq1$, so we have, $\frac{T_1+...+T_{k}}{k}\geq 1$ for all $ k$ , and this implies
%$$\frac{1}{\frac{T_1+...+T_{k-1}+1}{k-1}}\leq\frac{1}{1+\frac{1}{k-1}}\leq1$$
%Which means that by the Bounded Convergence Theorem, 
By the Strong Law of Large Numbers, with probability one,
$$\frac{k^{\zeta}}{\sigma_{k}^\zeta}=
\frac{1}{\big(\frac{T_1+T_2+\cdots+T_{k}+1}{k}\big)^\zeta}\rightarrow\frac{1}{(E(T_1))^\zeta}={p^\zeta}$$ as $k\rightarrow\infty$.
%$$=E\Big[\frac{1}{\big(\frac{T_1+T_2+\cdots+T_{k-1}+1}{k-1}\big)^\zeta}\Big]\frac{1}{(k-1)^\zeta}\rightarrow\frac{1}{(E(T_1))^\zeta}\frac{1}{(k-1)^\zeta}=\frac{p^\zeta}{(k-1)^\zeta}$$ as $k\rightarrow\infty$
Therefore,
$$\prod_{k=n_0}^{n}\alpha_k\leq e^{-\sum_{k=n_0}^{n}\frac{m \epsilon_0^2}{4\sigma_{n}^{\zeta}}}\leq \exp\Big(-\frac{m\epsilon_0^2}{4}\sum_{k=n_0}^{n}\frac{p^\zeta}{k^{\zeta}} \frac{\frac{1}{\sigma_{k}^\zeta}}{\frac{p^\zeta}{k^\zeta}}\Big)$$
Since $\displaystyle\frac{\frac{1}{\sigma_{k}^\zeta}}{\frac{p^\zeta}{k^\zeta}}\rightarrow1$, almost surely, as $k\rightarrow \infty$, we have that $\prod_{k=n_0}^{n}\alpha_k\rightarrow0$ a.s., if $\zeta\leq1$.
And, we conclude that, $Q_{\sigma_s}\cdots Q_{\sigma_{n-1}} \rightarrow\Pi$, a.s.,  for all $0 < \zeta\leq1$ and $0<p\leq1$.

\end{proof}

Note that for the special homogeneous case $\zeta=0$, the result is already known, and it was proved by Lagro et al.  \cite{LagroQuantumLimit} that the quantum Markov chain is convergent.
\cite{LagroQuantumLimit} did not use the compound Markov chain representation. It used the spectral theory of the density operators. So it does not contain this type of results for $\zeta =0$

\section{Convergence of density operators}\label{sec:convergence of density operators}
In this section we will prove convergence of density operators for $\zeta \le 1$ and $0 < p \le 1$. The homogeneous case $\zeta =0$ has been done in Lagro, Yang and Xiong \cite{LagroQuantumLimit} using spectral analysis of density operators. So  we will just focus on the case when $0< \zeta \le 1$. Let us recall a few definitions that will be  needed for this section. 

We consider 
$H=span\{|1\rangle,...,|m\rangle\}$ with $m\in\mathbb{N}$, and $U_n=e^{i\frac{G}{\sqrt{n^\zeta}}}$ where $G$ is an $m\times m$ self-adjoint matrix such that there exists $\epsilon_0>0$, $|G_{ij}|>\epsilon_0$ for all $i, j$. Let 
% with the following
$$ \Phi_{n}(\rho)=\sum_{i=0}^{m}A_iU_n\rho U_n^*A_i^* $$ and
\begin{eqnarray}\label{generalMC}
\rho_n=\Phi_{n}\cdots\Phi_{1}(\rho)
\end{eqnarray}
where $A_0=\sqrt{1-p}I$ and $A_i=\sqrt{p}|i\rangle \langle i|$, $1\leq i \leq m $, and let $\rho_0=|i\rangle \langle i|$ fixed, and $$P_n(i, j)=Tr\big(|j\rangle \langle j|\rho_n\big)$$
Let  $T_1,T_2,...$ be i.i.d. with distribution {\bfseries Geo($p$)}, and $\sigma_n=T_1+\cdots +T_n$. For a  fixed $t$, we let $n_t=\max\{n:\sigma_n\leq t\}$.
By Theorem \ref{th:Compond Markov Chain}, we have 
\begin{eqnarray}\label{eq:rep formula}
P_t(i, j)=Tr\big(|j\rangle \langle j|\rho_t\big)=E[Q_{\sigma_0}Q_{\sigma_1}...Q_{\sigma_{n_t-1}}W_{\sigma_{n_t}}(i,j)]
\end{eqnarray}
where $Q_{\sigma_{n-1}}(i,j)=|\langle j|U_{\sigma_{n-1}+T_n}\cdots U_{\sigma_{n-1}+1}|i\rangle|^2$ and $W_{\sigma_{n_t}}(i,j)=|\langle j|U_t\cdots U_{\sigma_n+1}|i\rangle|^2$.

We first prove convergence of probabilities as follows.
\begin{proposition}\label{generalEqui}
Suppose  there exists $\epsilon_0>0$ such that $|G_{ij}|>\epsilon_0$ for all $i, j$.
%Let $P_t(i, j) =P_t(j)$ with initial state $\rho_0=|i\rangle \langle i|$.}
	For all probability distribution $V$ and $j=1,2,...,m$, if $0<\zeta\leq1$ and $0<p\leq 1$, then
	$$\sum_{i}V_iP_t(i, j)\rightarrow \pi_j=\frac{1}{m},$$
	as $t\rightarrow \infty$.	
\end{proposition} 

Using the same argument as in the proof, under the same conditions of Proposition \ref{generalEqui}, we have 
\begin{corollary}\label{cor:general almost sure limit}
 Suppose  there exists $\epsilon_0>0$ such that $|G_{ij}|>\epsilon_0$ for all $i, j$. If $0<\zeta\leq1$ and $0<p\leq 1$, then for all $0 \le s < \infty$, we have
$$Q_{\sigma_s}Q_{\sigma_1}...Q_{\sigma_{n_t-1}}W_{\sigma_{n_t}} \to \Pi,$$
almost surely, as $t \to \infty$.
\end{corollary}

\begin{proof}[Proof of Proposition \ref{generalEqui}] Let $\Pi (i, j)=\pi_{j}=1/m$.
Observe that we can write $P_t(i, j)-\Pi (i, j)=$
$$E[Q_{\sigma_0}Q_{\sigma_1}...Q_{\sigma_{n_t-1}}W_{\sigma_{n_t}}-\Pi]=E[Q_{\sigma_0}Q_{\sigma_1}...Q_{\sigma_{n_t-1}}W_{\sigma_{n_t}}-Q_{\sigma_0}Q_{\sigma_1}...Q_{\sigma_{n_t-1}}+Q_{\sigma_0}Q_{\sigma_1}...Q_{\sigma_{n_t-1}}-\Pi]$$
$$=E[Q_{\sigma_0}Q_{\sigma_1}...Q_{\sigma_{n_t-1}}(W_{\sigma_{n_t}}-I\big)]+E\big[Q_{\sigma_0}Q_{\sigma_1}...Q_{\sigma_{n_t-1}}-\Pi\big]$$
Since $0<\zeta\leq1$, the second term tends to $0$ as $t\rightarrow \infty$ by Proposition  \ref{equi}.  Since $Q_{\sigma_0}Q_{\sigma_1}...Q_{\sigma_{n_t-1}} $ and  $W_{\sigma_{n_t}}-I $ are bounded,  by the Bounded Convergence Theorem, it is sufficient  to show that
\begin{eqnarray}\label{eq:W limit}
\lim_{t \to \infty} W_{\sigma_{n_t}}=I,  a.s.
\end{eqnarray}

$$W_{\sigma_{n_t}}(i,j)=|\langle j|U_t\cdots U_{\sigma_{n_t}+1}|i\rangle|^2=\big|\langle j|e^{\frac{iG}{\sqrt{t^\zeta}}}\cdots e^{\frac{iG}{\sqrt{(\sigma_{n_t}+1)^\zeta}}}|i\rangle\big|^2$$
$$=\big|\langle j|e^{iG\sum_{k=\sigma_{n_t}+1}^{t}\frac{1}{\sqrt{k^\zeta}}}|i\rangle\big|^2$$
Note that if we show that $|\langle j|e^{iG\sum_{k=\sigma_{n_t}+1}^{t}\frac{1}{\sqrt{k^\zeta}}}|i\rangle\big|^2\rightarrow 0$ a.s. for $i\neq j$, then $W_{\sigma_{n_t}}(i,j)\rightarrow \delta_{ij}$, a.s. and we will obtain the result.

Let's write $U_n=U_n^F+U_n^D$ for every $n$ where $U_n^F$ is the off-diagonal part and $U_n^D$ is the diagonal parts of $U_n$. So, for $i\neq j$,
$$|\langle j|e^{iG\sum_{k=\sigma_{n_t}+1}^{t}\frac{1}{\sqrt{k^\zeta}}}|i\rangle\big|^2=\Big|\sum_{l=\sigma_{N_{t}}+1}^{t}\langle j|U_t\cdots U_{l+1}U_{l}^FU_{l-1}^D\cdots U_{\sigma_{N_{t}}+1}^D|i\rangle\Big|^2$$
The above expansion is made by applying $U\cdots U(U^F+U^D)=U\cdots UU^F+U\cdots U(U^F+U^D)U^D=\cdots$, successively. By the Cauchy-Schwarz Inequality,
$$\leq (t-\sigma_{N_{t}})\sum_{l=\sigma_{N_{t}}+1}^{t}\big|\langle j|U_t\cdots U_{l+1}U_{l}^FU_{l-1}^D\cdots U_{\sigma_{N_{t}}+1}^D|i\rangle\Big|^2$$
By definition, we have $\displaystyle U_n=e^{\frac{iG}{\sqrt{n^\zeta}}}$, where $G$ is self-adjoint  $m$ by $m$ matrix and $|U_n(i, j)|\leq 1$ for all $i, j$,  the above is 
$$\leq (t-\sigma_{N_{t}})\sum_{l=\sigma_{N_{t}}+1}^{t}\big|\langle j|U_t\cdots U_{l+1}U_{l}^F|i\rangle U_{l-1}(i,i)\cdots U_{\sigma_{N_{t}}+1}(i,i)\big|^2$$
$$\leq (t-\sigma_{N_{t}})\sum_{l=\sigma_{N_{t}}+1}^{t}\big|\sum_{k\neq i}\langle j|U_t\cdots U_{l+1}|k\rangle U_{l}(k,i) U_{l-1}(i,i)\cdots U_{\sigma_{N_{t}}+1}(i,i)\big|^2$$
By the Cauchy-Schwarz Inequality again,
$$\leq (t-\sigma_{N_{t}})\sum_{l=\sigma_{N_{t}}+1}^{t}(m-1)\sum_{k\neq i}\big|\langle j|U_t\cdots U_{l+1}|k\rangle\big|^2 \big|U_{l}(k,i)\big|^2 \big|U_{l-1}(i,i)\cdots U_{\sigma_{N_{t}}+1}(i,i)\big|^2$$
Using the fact that
$\big|\langle j|U_t\cdots U_{l+1}|k\rangle\big|^2\leq 1$, and by Lemma \ref{le:lower bound semigroup}, $ \big|U_{l}(k,i)\big|^2\le 4 ||G||_{\infty}^2(\frac{1}{l^\zeta})$, for sufficiently large $l$, and  $$\displaystyle\big|U_{l-1}(i,i)\cdots U_{\sigma_{N_{t}}+1}(i,i)\big|^2\leq 1,$$ 
we conclude
$$(t-\sigma_{N_{t}})\sum_{l=\sigma_{N_{t}}+1}^{t}(m-1)^2 4 ||G||_{\infty}^2\big(\frac{1}{l^\zeta}\big)\leq (m-1)^2(t-\sigma_{N_{t}})^24 ||G||_{\infty}^2\Big(\frac{1}{{(\sigma_{N_{t}}+1)}^\zeta}\Big) $$
Now, it's enough to prove that $\displaystyle \frac{(t-\sigma_{N_{t}})^2}{(\sigma_{N_{t}})^\zeta}\leq \frac{T_{N_t+1}^2}{(T_1+\cdots + T_{N_t})^\zeta}\rightarrow 0$ a.s. The first inequality is trivial, and note that
$$\sum_{n=1}^{\infty}P\Big(\Big|\frac{T_n^2}{n^\zeta}\Big|>\epsilon\Big)=\sum_{n=1}^{\infty}P\Big(T_n^2>n^\zeta\epsilon\Big)=\sum_{n=1}^{\infty}P\Big(\frac{T_1^\frac{2}{\zeta}}{\epsilon^\frac{1}{\zeta}}>n\Big)\leq E\big(\frac{T_1^\frac{2}{\zeta}}{\epsilon^\frac{1}{\zeta}}\big)+1<\infty.$$
By the Borel-Cantelli Lemma,
$$P\Big(\Big|\frac{T_n^2}{n^\zeta}\Big|>\epsilon \text{ i.o. }\Big)=0$$
Therefore
$$\frac{T_n^2}{n^\zeta}\rightarrow 0 \text{ a.s.}$$
So, using the Strong Law of Large Numbers, we can conclude that
$$\frac{T_{n+1}^2}{(T_1+\cdots T_n)^\zeta}=\frac{\frac{T_{n+1}^2}{n^\zeta}}{(\frac{T_1+\cdots T_n}{n})^\zeta}\rightarrow \frac{0}{E(T_1)^\zeta} \text{ a.s. }$$

\end{proof}

Now we show our main result for the  convergence of density operators.
%{\color{red} Therefore, we have the most general convergence theorem for decoherent quantum Markov chains as follows,}
For $\zeta=0$, the following result has been obtained in \cite{LagroQuantumLimit} with more general conditions. Here we prove for the case $0 < \zeta\leq1$. 
\begin{theorem}\label{finalEqui}
	Suppose that $|G_{ij}|>\epsilon_0>0$ for all $i,j$. Suppose  $0 < \zeta\leq1$ and $0<p\leq1$.  Then for any initial density matrix $\rho_0$, we have 
\begin{eqnarray}
\rho_t\rightarrow\sum_{i=1}^{m}\pi_i\cdot|i\rangle\langle i| \label{eq:density limit result}
\end{eqnarray}
where $\pi_i=\frac{1}{m}$, for $i=1,...,m$.
\end{theorem}

\noindent
\begin{proof} [Proof of Theorem \ref{finalEqui}]
 Let  $0 < \zeta \le 1$ and $0< p \le 1$.

Step 1. We first consider the case $\rho_0=|i\rangle\langle i|$. 
By Proposition \ref{generalEqui}, we have  
\begin{eqnarray}
\lim_{ t \to \infty} Tr\big(|j\rangle \langle j|\rho_t\big)=\lim_{ t \to \infty} P_t(i, j)= \pi_j.
\end{eqnarray}
To prove (\ref{eq:density limit result}), it is then sufficient to show
\begin{eqnarray}\label{eq:off diagonal limit}
\lim_{ t \to \infty} Tr\big(|j\rangle \langle k|\rho_t\big)= 0,
\end{eqnarray}
for all $k \ne j$.

$$Tr\big(|k\rangle\langle j|\rho_t\big)=\langle j|\rho_t|k\rangle =\langle j|(\Phi_t\cdots \Phi_2\Phi_1|i\rangle \langle i|)|k\rangle$$
$$=\sum_{i_1,i_2,...,i_t=0}^{m}\langle j|A_{i_t}U_t\cdots A_{i_1}U_1|i\rangle\langle i|U_1^*A_{i_1}^*\cdots U_t^*A_{i_t}^*|k\rangle$$
$$=E\Big[\sum_{j_1,...,j_{N_t}=1}^{m}\langle j|U_t\cdots U_{\sigma_{N_t}+1}|j_{N_t}\rangle\langle j_{N_t}|U_{\sigma_{N_t}}\cdots U_{\sigma_{N_t-1}+1}|j_{N_t-1}\rangle\cdots$$
$$\cdots\langle j_1|U_{\sigma_1}\cdots U_2U_1|i\rangle\langle i|U_1^*U_2^*\cdots U_{\sigma_1}^*|j_1\rangle\cdots\langle j_{N_t-1}|U_{\sigma_{N_{t}-1}+1}^*\cdots U_{\sigma_{N_{t}}}^*|j_{N_t}\rangle \times$$
$$\times\langle j_{N_t}|U_{\sigma_{N_{t}}+1}^*\cdots U_{t}^*|k\rangle\Big]$$
$$=E\Big[\sum_{j_{N_t}=1}^{m}\langle j|U_t\cdots U_{\sigma_{N_t}+1}|j_{N_t}\rangle(Q_{\sigma_0}Q_{\sigma_1}\cdots Q_{\sigma_{N_t-1}})(i,j_{N_t})\times$$
$$\times\langle j_{N_t}|U_{\sigma_{N_{t}}+1}^*\cdots U_{t}^*|k\rangle\Big]$$
Let
$$ W(i,j,k):=\langle j|U_t\cdots U_{\sigma_{N_{t}}+1}|i\rangle\langle i|U_{\sigma_{N_{t}}+1}^*\cdots U_t^*|k\rangle$$
$$=\langle j|U_t\cdots U_{\sigma_{N_{t}}+1}|i\rangle\overline{\langle k|U_t\cdots U_{\sigma_{N_{t}}+1}|i\rangle}$$
So, we have that $$Tr\big(|k\rangle\langle j|\rho_t\big)=E\big[Q_{\sigma_0}Q_{\sigma_1}\cdots Q_{\sigma_{N_t-1}}W\big](i,j,k),$$
here for any matrix $Q$, we  define $$QW(i,j,l)=\sum_{l}Q(i,l)W(l,j,k)$$
Then 
$$|Q_{\sigma_0}Q_{\sigma_1}\cdots Q_{\sigma_{N_t-1}}W(i,j,k)|=\big|\sum_{l}Q_{\sigma_0}Q_{\sigma_1}\cdots Q_{\sigma_{N_t-1}}(i,l)W(l,j,k)\big|$$
$$\leq \sup_{l}||w(l,j,k)|\sum_{l}Q_{\sigma_0}Q_{\sigma_1}\cdots Q_{\sigma_{N_t-1}}(i,l)=\sup_{l}|w(l,j,k)|.$$
Since either $j$ or $k\neq l$, and $\zeta>0$, we may without loss of generality, assume that $j\neq l$. Then 
$$\big|\langle j|U_t\cdots U_{\sigma_{N_{t}+1}}|l\rangle\big|^2\rightarrow 0, a.s.$$ by (\ref{eq:W limit}). Since 
$$\big|\langle k|U_t\cdots U_{\sigma_{N_{t}}}|l\rangle\big|^2\leq 1,$$ we have 
$$\lim_{t \to \infty} \sup_{l}|w(l,j,k)|\rightarrow 0, a.s.$$
Therefore, (\ref{eq:off diagonal limit}) holds by the Dominated Convergence Theorem.

Step 2. Now we consider for general $\rho_0$.
For any $j, k$, by definition, 
$$Tr\big(|k\rangle\langle j|\rho_t\big)=\langle j|\rho_t|k\rangle =\langle j|(\Phi_t\cdots \Phi_2\Phi_1\rho_0)|k\rangle$$
$$=\sum_{i_1,i_2,...,i_t=0}^{m}\langle j|A_{i_t}U_t\cdots A_{i_1}U_1\rho_0U_1^*A_{i_1}^*\cdots U_t^*A_{i_t}^*|k\rangle$$
$$=E\Big[\sum_{j_1,...,j_{N_t}=1}^{m}\langle j|U_t\cdots U_{\sigma_{N_t}+1}|j_{N_t}\rangle\langle j_{N_t}|U_{\sigma_{N_t}}\cdots U_{\sigma_{N_t-1}+1}|j_{N_t-1}\rangle\cdots$$
$$\cdots\langle j_1|U_{\sigma_1}\cdots U_2U_1\rho_0U_1^*U_2^*\cdots U_{\sigma_1}^*|j_1\rangle\cdots\langle j_{N_t-1}|U_{\sigma_{N_{t}-1}+1}^*\cdots U_{\sigma_{N_{t}}}^*|j_{N_t}\rangle \times$$
$$\times\langle j_{N_t}|U_{\sigma_{N_{t}}+1}^*\cdots U_{t}^*|k\rangle\Big]$$
$$=E\Big[\sum_{j_{1}=1}^{m}\sum_{j_{N_t}=1}^{m}v_{j_1}\langle j|U_t\cdots U_{\sigma_{N_t}+1}|j_{N_t}\rangle(Q_{\sigma_1}\cdots Q_{\sigma_{N_t-1}})(j_1,j_{N_t})\times$$
$$\times\langle j_{N_t}|U_{\sigma_{N_{t}}+1}^*\cdots U_{t}^*|k\rangle\Big]$$
where $v_{j_1}=\langle j_1|U_{\sigma_1}\cdots U_2U_1\rho_0U_1^*U_2^*\cdots U_{\sigma_1}^*|j_1\rangle$.

So, we have that $$Tr\big(|k\rangle\langle j|\rho_t\big)=E\big[\sum_{j_{1}=1}^{m}v_{j_1}Q_{\sigma_1}\cdots Q_{\sigma_{N_t-1}}W\big](j_1,j,k).$$

If $j \ne k$, then 
$$|Q_{\sigma_1}\cdots Q_{\sigma_{N_t-1}}W(j_1,j,k)|=\big|\sum_{l}Q_{\sigma_1}\cdots Q_{\sigma_{N_t-1}}(i,l)W(l,j,k)\big|$$
$$\leq \sup_{l}||w(l,j,k)|\sum_{l}Q_{\sigma_1}\cdots Q_{\sigma_{N_t-1}}(i,l)=\sup_{l}|w(l,j,k)|.$$
Since either $j$ or $k\neq l$, and $\zeta>0$, suppose that $j\neq l$
$$\big|\langle j|U_t\cdots U_{\sigma_{N_{t}+1}}|l\rangle\big|^2\rightarrow 0$$ by (\ref{eq:W limit}) and
$$\big|\langle k|U_t\cdots U_{\sigma_{N_{t}}}|l\rangle\big|^2\leq 1,$$ we have 
$$\lim_{t \to \infty} \sup_{l}|w(l,j,k)|\rightarrow 0, a.s.$$

 If $j = k$, then 
$$\lim_{t \to \infty} Q_{\sigma_1}\cdots Q_{\sigma_{N_t-1}}W(j_1,j,k)=\pi_{j}, a.s.,$$
by Corollary \ref{cor:general almost sure limit}. 
Since $v_j$ is a probability distribution, by the Dominated Convergence Theorem, we have 
$$Tr\big(|k\rangle\langle j|\rho_t\big)=E\big[\sum_{j_{1}=1}^{m}v_{j_1}Q_{\sigma_1}\cdots Q_{\sigma_{N_t-1}}W\big](j_1,j,k) \to \pi_j.$$ 

We have thus proved the theorem

\end{proof}

\begin{remark}
	Using a similar argument, we can show that the conclusion of Theorem \ref{finalEqui} also holds for Example \ref{time-inho HadamardMC}, i.e., if $0<\zeta\leq 1$ and $0 < p \leq 1$, then
	$$ \rho_t\rightarrow\sum_{i=1}^{2}\pi_i\cdot|i\rangle\langle i|$$
	where $\pi_i=\frac{1}{2}$, for $i=1,2$.
\end{remark}
\section{Critical behaviors and simulations }\label{sec:simulations}
In this section, we consider critical behaviors and discuss their numerical  results. Our main observations are the following,
\begin{enumerate}

\item For $p =0$, our quantum Markov chain exhibits a time-inhomogeneous periodic behavior for $\zeta \le 2$ and non-periodic for $\zeta >2$. At time $n$, let $T_n$ denote the first period after time $n$ for fixed $\lambda$ and $\zeta$, we have for large $n$
 
\begin{itemize}

	\item If $\zeta <2$, 
	
	\begin{eqnarray}
	T_n \approx \Big[\big(1-\frac{\zeta}{2}\big)\frac{2\pi}{\lambda_1-\lambda_2}+n^{1-\frac{\zeta}{2}}\Big]^{\frac{1}{1-\frac{\zeta}{2}}}-n,
	\label{eq:PureQuantumPeriod01} 
	\end{eqnarray}
	\item If $\zeta=2$,
	
	\begin{eqnarray}
	T_n \approx n\cdot(e^{\frac{2\pi}{\lambda_1-\lambda_2}}-1),
	\label{eq:PureQuantumPeriod02} 
	\end{eqnarray}
\end{itemize}

where $\lambda_1>\lambda_2$ are two eigenvalues of the matrix $G=\left[
\begin{array}{ c c  }
	\lambda & \lambda \\
	\lambda & \lambda
\end{array} \right]$.
Moreover, let $T_n^{(\lambda)}$ denote the first period after time $n$ for fixed $\zeta$ depending on $\lambda$, the relation

\begin{eqnarray}
	T_n^{(\lambda)}\sim \frac{T_n^{(1)}}{\lambda}
	\label{eq:PureQuantumPeriod03} 
\end{eqnarray}
can be directly deduced by Formulas (\ref{eq:PureQuantumPeriod01}) and (\ref{eq:PureQuantumPeriod02}) using Taylor series expansion.

\item For small $p >0$, the time-inhomogeneous periodicity behaves the same manner as the case $p=0$.

\item For small $p >0$, the amplitudes of each period decays exponentially as $n$ gets large. The exponential decay rate is  approximately $p/2$, independent of the choices of $\zeta \le 2$ and $\lambda >0$.

\item For $p \approx 1$, the decay to limiting distribution is exponentially fast if $\zeta $ is small, and it is power law decay if $\zeta $ is large. The critical point $\zeta _c$ appears in the interval $[0.6, 0.7]$.

\item For $p>0$, the critical phase transition occurs when $\zeta_c=1$ for the decoherent Markov chain.

\item Instead of fully connected graph we considered, simulation results show that the decoherent system converges equilibrium limit with $N$ dimensional cyclic graph $G$. For $p=0$, time-inhomogeneous periodicity is not predictable although some kind of periodicity was observed.

\end{enumerate} 

In what follows, we will discuss the numerical results for the above observations assuming that all the entries of the matrix $G$ are $\lambda>0$ for simplicity.

\subsection{Pure quantum case, $p=0$, critical behavior at $\zeta_o=2$}\label{section:timInhomoPeriod}

Let us consider the pure quantum case when $p=0$.  For $2$ dimensional case $i,j=1,2$, recall Definition \ref{InhomoDecohOperator} and Equation \ref{InhomoDecohProbability}, the probability starting from the state $|i\rangle$ and getting the state $|j\rangle$ after $n$ steps can be deduced to

\begin{eqnarray}
P(i,j)=\big| \langle j |U_n\cdots U_1|i\rangle \big|^2,
\label{eq:PureQuantumProbability}
\end{eqnarray}
using the fact $U_n=e^{i\frac{G}{\sqrt{n^\zeta}}}$, where $G=$ is a self-adjoint matrix with non-zero entries,
$$P(i,j)=\big| \langle j |e^{iG\sum_{k=1}^{n}\frac{1}{\sqrt{k^\zeta}}}|i\rangle \big|^2=\big| \langle j |B  \left[
\begin{array}{ c c  }
e^{i\lambda_1\sum_{k=1}^{n}\frac{1}{\sqrt{k^\zeta}}} & 0 \\
0 & e^{i\lambda_2\sum_{k=1}^{n}\frac{1}{\sqrt{k^\zeta}}}
\end{array} \right]   B^*|i\rangle \big|^2,$$
where $\lambda_1$ and $\lambda_2$ are eigenvalues of $G=\left[
\begin{array}{ c c  }
	\lambda & \lambda \\
	\lambda & \lambda
\end{array} \right]$ with $\lambda_1>\lambda_2$, and B the associated orthogonal matrix. Let $$\Lambda_n=\left[
\begin{array}{ c c  }
e^{i\lambda_1\sum_{k=1}^{n}\frac{1}{\sqrt{k^\zeta}}} & 0 \\
0 & e^{i\lambda_2\sum_{k=1}^{n}\frac{1}{\sqrt{k^\zeta}}}
\end{array} \right],$$ and we have
$$P_n(i,j)= \Big|\sum_{l=1}^{2}\langle j|B|l\rangle\langle l| \Lambda_n |l\rangle \langle l|B^*|i\rangle       \Big|^2 = \Big|\sum_{l=1}^{2}\langle j|B|l\rangle\cdot\Lambda_n(l,l)\cdot \langle l|B^*|i\rangle       \Big|^2$$
$$=\Big|\langle j|B|1\rangle\cdot\Lambda_n(1,1)\cdot \langle 1|B^*|i\rangle  + \langle j|B|2\rangle\cdot\Lambda_n(2,2)\cdot \langle 2|B^*|i\rangle     \Big|^2$$
$$=\Big|\langle j|B|1\rangle\cdot e^{i\lambda_1\sum_{k=1}^{n}\frac{1}{\sqrt{k^\zeta}}}\cdot \langle 1|B^*|i\rangle  + \langle j|B|2\rangle\cdot e^{i\lambda_2\sum_{k=1}^{n}\frac{1}{\sqrt{k^\zeta}}}\cdot \langle 2|B^*|i\rangle     \Big|^2.$$
Suppose that $i=j=1$, with $B=\left[
\begin{array}{ c c  }
b_{11} & b_{12} \\
b_{21} & b_{22}
\end{array} \right]$, we obtain
$$P_n(1,1)=\Big|\langle 1|B|1\rangle\cdot e^{i\lambda_1\sum_{k=1}^{n}\frac{1}{\sqrt{k^\zeta}}}\cdot \langle 1|B^*|1\rangle  + \langle 1|B|2\rangle\cdot e^{i\lambda_2\sum_{k=1}^{n}\frac{1}{\sqrt{k^\zeta}}}\cdot \langle 2|B^*|1\rangle     \Big|^2$$
$$=\Big|B(1,1) \cdot e^{i\lambda_1\sum_{k=1}^{n}\frac{1}{\sqrt{k^\zeta}}}\cdot B^*(1,1)  + B(1,2) \cdot e^{i\lambda_2\sum_{k=1}^{n}\frac{1}{\sqrt{k^\zeta}}}\cdot B^*(2,1)     \Big|^2$$
$$=\Big|b_{11} \cdot e^{i\lambda_1\sum_{k=1}^{n}\frac{1}{\sqrt{k^\zeta}}}\cdot \bar{b}_{11}  + b_{12} \cdot e^{i\lambda_2\sum_{k=1}^{n}\frac{1}{\sqrt{k^\zeta}}}\cdot \bar{b}_{12}     \Big|^2=\Big| |b_{11}|^2 \cdot e^{i\lambda_1\sum_{k=1}^{n}\frac{1}{\sqrt{k^\zeta}}} + |b_{12}|^2 \cdot e^{i\lambda_2\sum_{k=1}^{n}\frac{1}{\sqrt{k^\zeta}}} \Big|^2$$
$$=\Big(|b_{11}|^2 \cdot e^{i\lambda_1\sum_{k=1}^{n}\frac{1}{\sqrt{k^\zeta}}} + |b_{12}|^2 \cdot e^{i\lambda_2\sum_{k=1}^{n}\frac{1}{\sqrt{k^\zeta}}}\Big)\Big(|b_{11}|^2 \cdot e^{-i\lambda_1\sum_{k=1}^{n}\frac{1}{\sqrt{k^\zeta}}} + |b_{12}|^2 \cdot e^{-i\lambda_2\sum_{k=1}^{n}\frac{1}{\sqrt{k^\zeta}}}\Big)$$
$$=|b_{11}|^4+|b_{11}|^2|b_{12}|^2 e^{i(\lambda_1-\lambda_2)\sum_{k=1}^{n}\frac{1}{\sqrt{k^\zeta}}}+|b_{12}|^4+|b_{12}|^2|b_{11}|^2 e^{i(\lambda_2-\lambda_1)\sum_{k=1}^{n}\frac{1}{\sqrt{k^\zeta}}}$$
$$=|b_{11}|^4+|b_{12}|^4+2|b_{12}|^2|b_{11}|^2\cos\big((\lambda_1-\lambda_2)\sum_{k=1}^{n}\frac{1}{\sqrt{k^\zeta}}\big)$$
$$=1-2|b_{12}|^2|b_{11}|^2\Big[1-\cos\big((\lambda_1-\lambda_2)\sum_{k=1}^{n}\frac{1}{\sqrt{k^\zeta}}\big)\Big]$$
Finally, we conclude
$$P_n(1,1)= 1-2|b_{12}|^2|b_{11}|^2\Big[1-\cos\big((\lambda_1-\lambda_2)\sum_{k=1}^{n}\frac{1}{\sqrt{k^\zeta}}\big)\Big],$$
and
$$P_n(2,1)=1-\Big(1-2|b_{12}|^2|b_{11}|^2\Big[1-\cos\big((\lambda_1-\lambda_2)\sum_{k=1}^{n}\frac{1}{\sqrt{k^\zeta}}\big)\Big]\Big)$$
$$=2|b_{12}|^2|b_{11}|^2\Big[1-\cos\big((\lambda_1-\lambda_2)\sum_{k=1}^{n}\frac{1}{\sqrt{k^\zeta}}\big)\Big].$$
We observe from above deduction that the probability at $n$ depends on the eigenvalues and eigenvectors of $G$ and $\zeta$. In our case, $B$ has non-zero entries with $\lambda_1>\lambda_2$, and the probabilities in $P_n$ has periodicity property if $\zeta\leq 2$, the series $\displaystyle\sum_{k=1}^{n}\frac{1}{\sqrt{k^\zeta}}$ is divergent as $n$ tends to infinity. On the other hand, if $\zeta>2$, $P_n$ converges, and the limit depends on the matrix $G$ and the value of $\zeta$.

\begin{figure}[h!]%
	\centering
	\subfigure[$\zeta=1$]{%
		\label{fig:QMC01a}%
		\includegraphics[width=1.9in,height=1.9in]{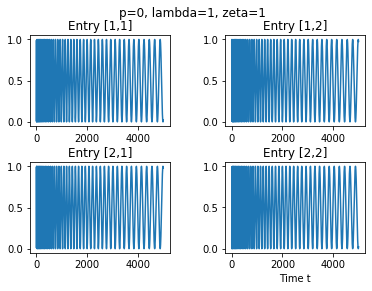}}%
	\hfill
	\subfigure[$\zeta=1.5$]{%
		\label{fig:QMC01b}%
		\includegraphics[width=1.9in,height=1.9in]{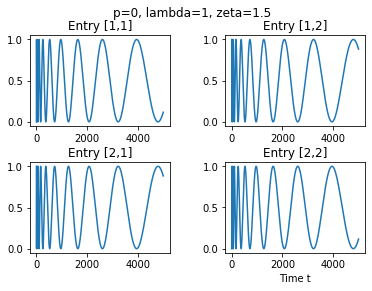}}%
	\hfill
	\subfigure[$\zeta=2$]{%
		\label{fig:QMC01c}%
		\includegraphics[width=1.9in,height=1.9in]{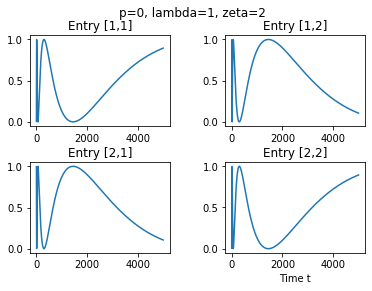}}%	
	\caption{Simulation results for $P_n$ fixing $p=0$ and $\lambda=1$ with $t=50000$}
	\label{fig:QMC01}
\end{figure}
\begin{figure}[h!]%
	\centering
	\subfigure[$\lambda=0.5$]{%
		\label{fig:QMC02a}%
		\includegraphics[width=1.9in,height=1.9in]{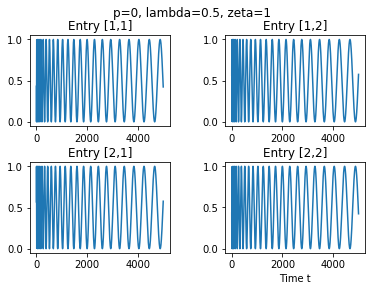}}%
	\hfill
	\subfigure[$\lambda=1.5$]{%
		\label{fig:QMC02b}%
		\includegraphics[width=1.9in,height=1.9in]{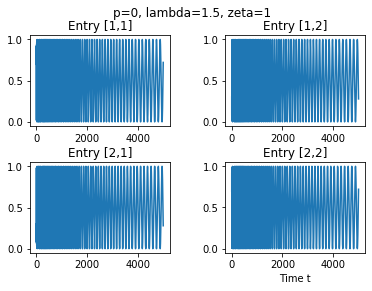}}%
	\hfill
	\subfigure[$\lambda=2$]{%
		\label{fig:QMC02c}%
		\includegraphics[width=1.9in,height=1.9in]{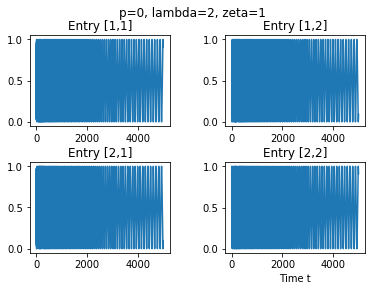}}%	
	\caption{Simulation results for $P_n$ fixing $p=0$ and $\zeta=1$ with $t=50000$}
	\label{fig:QMC02}
\end{figure} 

Figures \ref{fig:QMC01} and \ref{fig:QMC02} show the periodicity behavior of the $2 \times 2$ matrix $P_n$ by fixing $\zeta$ and $\lambda$ with respectively. We observe from Figure \ref{fig:QMC01} that with fixed $\lambda$, the periodicity of $P_n$ decreases if $\zeta$ increases. On the other hand, Figure \ref{fig:QMC02} illustrates the phenomenon that with fixed $\zeta$, the periodicity of $P_n$ increases if $\lambda$ increases.

To analyze the periodicity of $P_n$ when $\zeta\leq 2$, we use the fact that $\cos(\cdot)$ has period $2\pi$, and let $T_n$ be the period for fixed $n$.
$$(\lambda_1-\lambda_2)\sum_{k=n+1}^{n+T_n}\frac{1}{\sqrt{k^\zeta}}=2\pi.$$
Asymptotically, we have for large $n$
$$\int_{n}^{n+T_n}\frac{1}{\sqrt{x^\zeta}}dx \approx \frac{2\pi}{\lambda_1-\lambda_2}.$$
If $\zeta<2$, it can be deduced to
$$\Big[\frac{x^{1-\frac{\zeta}{2}}}{1-\frac{\zeta}{2}}\Big]_{x=n}^{x=n+T_n}\approx\frac{2\pi}{\lambda_1-\lambda_2}\implies \frac{(n+T_n)^{1-\frac{\zeta}{2}}}{1-\frac{\zeta}{2}}-\frac{n^{1-\frac{\zeta}{2}}}{1-\frac{\zeta}{2}}\approx\frac{2\pi}{\lambda_1-\lambda_2},$$
then, we have
$$(n+T_n)^{1-\frac{\zeta}{2}}-n^{1-\frac{\zeta}{2}}\approx(1-\frac{\zeta}{2})\frac{2\pi}{\lambda_1-\lambda_2},$$
and
\begin{eqnarray}
T_n \approx \Big[\big(1-\frac{\zeta}{2}\big)\frac{2\pi}{\lambda_1-\lambda_2}+n^{1-\frac{\zeta}{2}}\Big]^{\frac{1}{1-\frac{\zeta}{2}}}-n.
\label{eq:PureQuantumPeriod1} 
\end{eqnarray}
If $\zeta=2$, we have
$$\Big[\ln|x|\Big]_{x=n}^{x=n+T_n}\approx\frac{2\pi}{\lambda_1-\lambda_2} \implies \ln|n+T_n|-\ln|n|\approx\frac{2\pi}{\lambda_1-\lambda_2},$$
then,
\begin{eqnarray}
T_n \approx n\cdot(e^{\frac{2\pi}{\lambda_1-\lambda_2}}-1).
\label{eq:PureQuantumPeriod2} 
\end{eqnarray}

Figure \ref{fig:QMC03} compares the periods generated by simulation and Formula (\ref{eq:PureQuantumPeriod1}) with different values of $\lambda$ and $\zeta$. As results, we observe that the ratios between the simulated periods and theoretical periods are $1$ as $n$ is large, which numerically proved the asymptotic formulas discussed.

\begin{figure}[h!]%
	\centering
	\subfigure[$\lambda=1$ and $\zeta=0.8$]{%
		\label{fig:QMC03a}%
		\includegraphics[width=1.4in,height=2.1in]{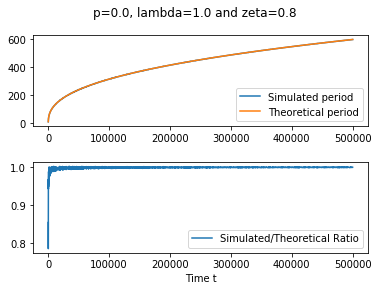}}%
	\hfill
	\subfigure[$\lambda=1$ and $\zeta=1.2$]{%
		\label{fig:QMC03b}%
		\includegraphics[width=1.4in,height=2.1in]{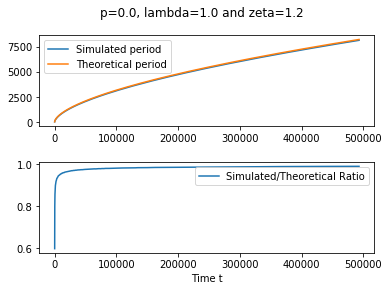}}%
	\hfill
	\subfigure[$\lambda=0.5$ and $\zeta=1.2$]{%
		\label{fig:QMC03c}%
		\includegraphics[width=1.4in,height=2.1in]{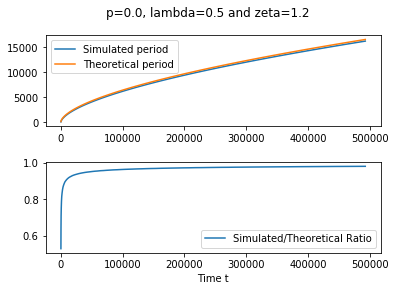}}%
	\hfill
	\subfigure[$\lambda=1.5$ and $\zeta=1.2$]{%
		\label{fig:QMC03d}%
		\includegraphics[width=1.4in,height=2.1in]{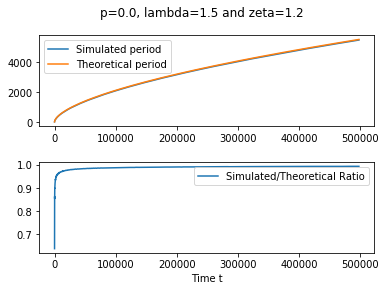}}%	
	\caption{Simulated vs theoretical periods for $P_n$ with $t=50000$}
	\label{fig:QMC03}
\end{figure}

%Additionally, from Figure \ref{fig:QMC01} and \ref{fig:QMC02}
\subsection{Periodicity and decay rates for decoherent quantum Markov chain, $p>0$}
Previously, we have numerically shown the periodic behavior for the pure quantum Markov chains using Formulas (\ref{eq:PureQuantumPeriod1}), (\ref{eq:PureQuantumPeriod2}). For $p>0$, despite the equilibrium convergence property proved in Theorem \ref{finalEqui}, we numerically analyze the periodicity and convergence rates to equilibrium for decoherent Markov chains.

\begin{figure}[h!]%
	\centering
	\subfigure[$\lambda= 0.5$ and $\zeta=0.4$]{%
		\label{fig:QMC04a}%
		\includegraphics[width=1.9in,height=1.9in]{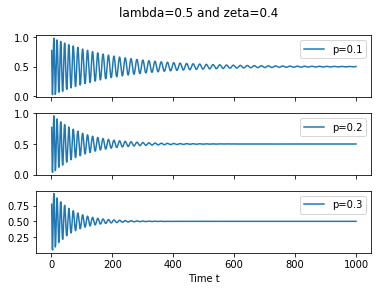}}%
	\hfill
	\subfigure[$\lambda= 0.5$ and $\zeta=0.6$]{%
		\label{fig:QMC04b}%
		\includegraphics[width=1.9in,height=1.9in]{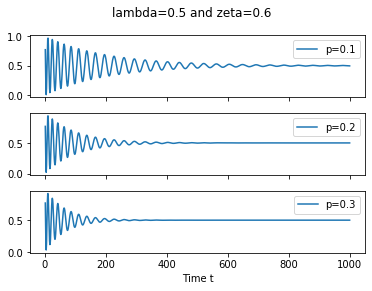}}%
	\hfill
	\subfigure[$\lambda= 0.5$ and $\zeta=0.8$]{%
		\label{fig:QMC04c}%
		\includegraphics[width=1.9in,height=1.9in]{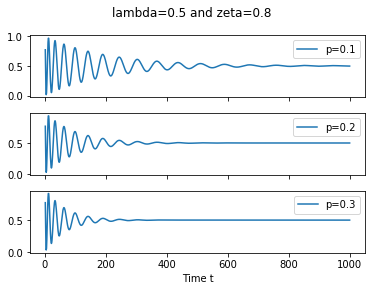}}%	
	\caption{Probabilities for $P_n(1,1)$ with $t=1000$}
	\label{fig:QMC04}
\end{figure}

Regarding the periodicity of the decoherent system, if $p$ is close to $0$, the probability to make measurements is low, the system maintains more quantum behavior by intuition. Conversely, the system tends to have less quantum coherence while $p$ increases. Simulation results not only accord with above observation, but also demonstrate that only small $p$'s conserve notable quantum periodic behavior (see Figure \ref{fig:QMC04}). Furthermore, the results show that for larger $\zeta$'s, fixed $p$'s and $\lambda$'s, the periodicity decreases similarly as the coherent case.

\begin{figure}[h!]%
	\centering
	\subfigure[$\lambda= 0.5$ and $\zeta=0.4$]{%
		\label{fig:QMC05a}%
		\includegraphics[width=1.9in,height=1.9in]{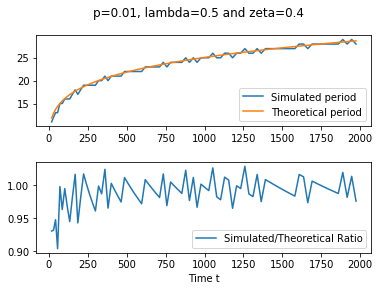}}%
	\hfill
	\subfigure[$\lambda= 0.5$ and $\zeta=0.6$]{%
		\label{fig:QMC05b}%
		\includegraphics[width=1.9in,height=1.9in]{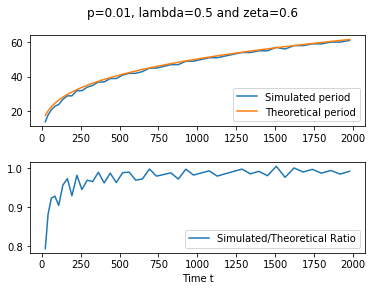}}%
	\hfill
	\subfigure[$\lambda= 0.5$ and $\zeta=0.8$]{%
		\label{fig:QMC05c}%
		\includegraphics[width=1.9in,height=1.9in]{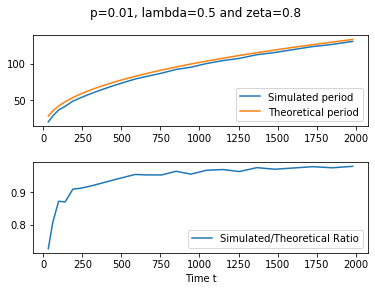}}%	
	\caption{Simulated vs pure quantum theoretical periods for $P_n$ with $p=0.01$ with $t=2000$}
	\label{fig:QMC05}
\end{figure}
\begin{figure}[h!]%
	\centering
	\subfigure[$\lambda= 0.5$ and $\zeta=0.4$]{%
		\label{fig:QMC06a}%
		\includegraphics[width=1.9in,height=1.9in]{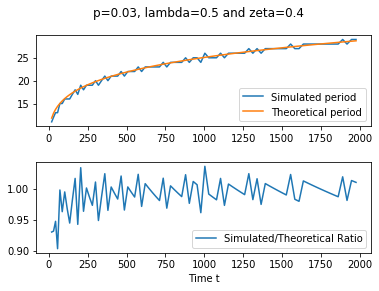}}%
	\hfill
	\subfigure[$\lambda= 0.5$ and $\zeta=0.6$]{%
		\label{fig:QMC06b}%
		\includegraphics[width=1.9in,height=1.9in]{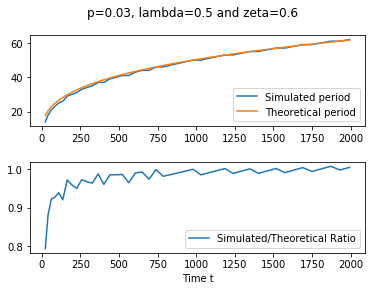}}%
	\hfill
	\subfigure[$\lambda= 0.5$ and $\zeta=0.8$]{%
		\label{fig:QMC06c}%
		\includegraphics[width=1.9in,height=1.9in]{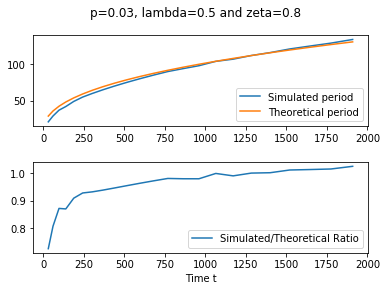}}%	
	\caption{Simulated vs pure quantum theoretical periods for $P_n$ with $p=0.03$ with $t=2000$}
	\label{fig:QMC06}
\end{figure} 

We also calculate the ratios between the simulated periodicity and the theoretical periodicity from Formula (\ref{eq:PureQuantumPeriod1}) fixing different values of $p$ and $\zeta$. As results, Figures \ref{fig:QMC05} and \ref{fig:QMC06} show that the ratios are asymptotically $1$. In other words, even for $p>0$, the decoherent system maintains the same periodicity as in pure quantum case for small $p$.

In order to analyze the convergence rates for small $p$ and different values of $\zeta$'s and $\lambda$'s, we fit the local maximums of the $P_n(1,1)$ in non-linear regression with exponential decay model: $ce^{-rt}-\frac{1}{2}$, and rational decay model: $ct^{-r}-\frac{1}{2}$ for $p$ near $0$. Since for every choice of $\lambda$ and $\zeta$, the number of local maximums are different, we consider the adjusted $R^2$ coefficients to perform the comparison. Consequently, for all parameters mentioned above, not only the adjusted $R^2$ of exponential models with decay rate $r$ are greater than rational models' adjusted $R^2$, but also they are all greater than $0.9$. Even more, we notice that the linear relationship between $r$ and $p$ when $p$ tends to $0$ for all $0\leq\zeta\leq1$, which is, the convergence rates are proportional to $p$ and they are independent of $\zeta$ and $\lambda$ as $p$ tends to $0$.

More specifically, Tables \ref{tab:QMC01}, \ref{tab:QMC02} and \ref{tab:QMC03} show the convergence rates obtained from exponential decay non-linear regression model for $\lambda= 0.2$, $0.35$ and $0.5$, and we numerically conclude that the decay rates $r \sim \frac{p}{2}$ as $p$ approaches $0$. In other words, when environmental interaction is extremely low, the convergence to equilibrium does not depend on the parameters $\zeta$ and $\lambda$.

\begin{table}[h!]
	\centering%\small
	\footnotesize
	\begin{tabular}{ |c||c|c|c|c|c|c|c|c|c|c|} 
		\hline
		$p  \  \backslash  \  \zeta$ & $0.1$ & $0.2$ & $0.3$ & $ 0.4$  & $ 0.5$ & $0.6$ & $0.7$ & $0.8$ & $0.9$ & $1$ \\ 
		\hline
		$0.005$ & 0.0025 & 0.0025 & 0.0025 & 0.0025   & 0.0025 & 0.0025 & 0.0025 & 0.0025 & 0.0025 & 0.0025  \\ 
		
		$0.01$ & 0.005 & 0.005 & 0.005 & 0.005   & 0.005 & 0.005 & 0.005 & 0.005 & 0.005 & 0.0051 \\  
		
		$0.015$ & 0.0075 & 0.0075 & 0.0075 & 0.0077   & 0.0077 & 0.0077 & 0.0077 & 0.0077 & 0.0076 & 0.0077 \\  
		
		$0.02$ & 0.0101 & 0.0101 & 0.0101 & 0.0101   & 0.0101 & 0.0102 & 0.0102 & 0.0101 & 0.0102 & 0.0103 \\  
		
		$0.025$ & 0.0126 & 0.0126 & 0.0127 & 0.0127  & 0.0127 & 0.0127 & 0.0128 & 0.0127 & 0.0128 & 0.0129 \\  
		
		%$0.03$ & 0.0151 & 0.0153 & 0.015 & 0.0152   & 0.0152 & 0.015 & 0.015 &  0.0152  &  0.0153  & 0.0151 \\  
		\hline
	\end{tabular}
	\caption{Estimated convergence rates $r$ for $\lambda=0.2$ by non-linear regression model: $ce^{-rt}$}
	\label{tab:QMC01}

\end{table}
\begin{table}[h!]
	\centering%\small
	\footnotesize
	\begin{tabular}{ |c||c|c|c|c|c|c|c|c|c|c|} 
		\hline
		$p  \  \backslash  \  \zeta$ & $0.1$ & $0.2$ & $0.3$ & $ 0.4$  & $ 0.5$ & $0.6$ & $0.7$ & $0.8$ & $0.9$ & $1$ \\ 
		\hline
		$0.005$ & 0.0025 & 0.0025 & 0.0025 & 0.0025   & 0.0025 & 0.0025 & 0.0025 & 0.0025 & 0.0025 & 0.0025  \\ 
		
		$0.01$ & 0.005 & 0.005 & 0.005 & 0.005   & 0.005 & 0.005 & 0.005 & 0.005 & 0.005 & 0.005 \\  
		
		$0.015$ & 0.0075 & 0.0075 & 0.0075 & 0.0075   & 0.0075 & 0.0075 & 0.0076 & 0.0076 & 0.0076 & 0.0076 \\  
		
		$0.02$ & 0.0101 & 0.01 & 0.01 & 0.01   & 0.0101 & 0.0101 & 0.0101 & 0.0101 & 0.0102 & 0.0102 \\  
		
		$0.025$ & 0.0126 & 0.0124 & 0.0125 & 0.0125  & 0.0126 & 0.0126 & 0.0127 & 0.0127 & 0.0128 & 0.0126 \\  
		\hline
		
	\end{tabular}
	\caption{Estimated convergence rates $r$ for $\lambda=0.35$ by non-linear regression model: $ce^{-rt}$}
	\label{tab:QMC02}

\end{table}
\begin{table}[h!]
	\centering%\small
	\footnotesize
	\begin{tabular}{ |c||c|c|c|c|c|c|c|c|c|c|} 
		\hline
		$p  \  \backslash  \  \zeta$ & $0.1$ & $0.2$ & $0.3$ & $ 0.4$  & $ 0.5$ & $0.6$ & $0.7$ & $0.8$ & $0.9$ & $1$ \\ 
		\hline
		$0.005$ & 0.0025 & 0.0025 & 0.0025 & 0.0025   & 0.0025 & 0.0025 & 0.0025 & 0.0025 & 0.0025 & 0.0025  \\ 
		
		$0.01$ & 0.005 & 0.005 & 0.005 & 0.005   & 0.005 & 0.005 & 0.005 & 0.005 & 0.005 & 0.005 \\  
		
		$0.015$ & 0.0075 & 0.0075 & 0.0075 & 0.0075   & 0.0075 & 0.0075 & 0.0075 & 0.0075 & 0.0076 & 0.0075 \\  
		
		$0.02$ & 0.0101 & 0.0101 & 0.01 & 0.0101   & 0.0101 & 0.01 & 0.01 & 0.01 & 0.0101 & 0.0101 \\  
		
		$0.025$ & 0.0126 & 0.0127 & 0.0125 & 0.0127  & 0.0127 & 0.0125 & 0.0125 & 0.0126 & 0.0127 & 0.0125 \\  
		
		%$0.03$ & 0.0151 & 0.0153 & 0.015 & 0.0152   & 0.0152 & 0.015 & 0.015 &  0.0152  &  0.0153  & 0.0151 \\  
		\hline
	\end{tabular}
	\caption{Estimated convergence rates $r$ for $\lambda=0.5$ by non-linear regression model: $ce^{-rt}$}
	\label{tab:QMC03}
\end{table}

On the other hand, we similarly fit the probabilities $P_n(1,1)$ with $n=200$ in non-linear regression with exponential decay model: $ce^{-rt}-\frac{1}{2}$, and rational decay model: $ct^{-r}-\frac{1}{2}$ for different values of $\lambda$, $\zeta$ and $p$ close to $1$ where the probability of measurement is high. In this case, we compare the $R^2$ coefficients for above parameters since there have no local maximums for $p$ close to $1$. Tables \ref{tab:QMC04}, \ref{tab:QMC05} and \ref{tab:QMC06} show that for $\lambda = 0.2$, $0.35$, $0.5$ and $p$ close to $1$, exponential decay model fits better for $0<\zeta\leq 0.6$ with  $R^2$ coefficients greater than $0.9$ while rational decay model fits better for $0.7 <\zeta\leq 1 $ with $R^2>0.9$. Therefore, not only we notice that the decay behavior is very different than when $p$ is close to 0, but also there is a phase of transition $\zeta_o$ such that $0.6\leq\zeta\leq0.7$ for $p\sim 1$.

Indeed, we also observe from these tables that the convergence rates decrease in function of $p$ when $p$ tends to $0$, and the convergence rates increase when $p$ tends to $1$.

\begin{table}[h!]
	\centering\footnotesize
	
	\begin{tabular}{ |c||c|c|c|c|c|c|c|c|c|c|} 
		\hline
		$p  \  \backslash  \  \zeta$ & $0.1$ & $0.2$ & $0.3$ & $ 0.4$  & $ 0.5$ & $0.6$ & $0.7$ & $0.8$ & $0.9$ & $1$ \\ 
		\hline
		
		$0.7$ &  &  & 1 &    &  &  &  &  & 1 & 2 \\ 
		
		$0.8$ & 1 &  & 1 & 1   & 1 & 1 & 2 & 2 & 2 & 2 \\  
		
		$0.9$ & 1 & 1 &  & 1   & 1 & 2 & 2 & 2 & 2 & 2 \\  
		
		$1$ & 1 & 1 & 1 & 1  & 1 & 1 & 2 & 2 & 2 & 2 \\  
		
		\hline
	\end{tabular}

	\

	\centerline{Model Fitness: 1 for Exponential decay model and 2 for Rational decay model}
	
	\
	
	\begin{minipage}{.4\linewidth}
		\footnotesize
		
		\begin{tabular}{ |c||c|c|c|c|} 
			\hline
			$\zeta  \  \backslash  \  p $ & $0.7$ & $0.8$ & $0.9$ & $1$ \\ 
			\hline
			
			$0.1$ &   & 1.33 & 0.84 & 0.54 \\
			
			$0.2$ &   &  & 0.74 & 0.47 \\
			
			$0.3$ &  1.65 & 1.06 &  & 0.40 \\
			
			$0.4$ &   & 0.94 & 0.56 & 0.34 \\
			
			$0.5$ &   & 0.83 & 0.47 & 0.27 \\
			
			$0.6$ &   & 0.72 &  & 0.21 \\
			
			$0.7$ &   &  &  &  \\  
			
			$0.8$ & &  &  &  \\  
			
			$0.9$ & 0.84 &  &  &  \\  
			
			$1$ &  &  &  &  \\  
			
			\hline
		\end{tabular}
	
	\
	
		\centerline{Exponential Decay Rates}
	\end{minipage}\ \ \ \ \ \
	\begin{minipage}{.4\linewidth}
		\footnotesize
		
		\begin{tabular}{ |c||c|c|c|c|} 
			\hline
			$\zeta  \  \backslash  \  p $ & $0.7$ & $0.8$ & $0.9$ & $1$ \\ 
			\hline
			
			$0.1$ &   &  &  &  \\
			
			$0.2$ &   &  &  &  \\
			
			$0.3$ &   &  &  &  \\
			
			$0.4$ &   &  &  &  \\
			
			$0.5$ &   &  &  &  \\
			
			$0.6$ &   &  & 1.11 &  \\
			
			$0.7$ &  & 1.32 & 1.01 & 0.83 \\  
			
			$0.8$ &  & 1.18 & 0.89 & 0.72 \\  
			
			$0.9$ &  & 1.02 & 0.75 & 0.6 \\  
			
			$1$ & 1.4 & 0.84 & 0.61 & 0.47 \\  
			
			\hline
		\end{tabular}
	
	\
	
		\centerline{Rational Decay Rates}
	\end{minipage}
\caption{Estimated convergence rates $r$ for $\lambda=0.5$ by non-linear regression model}
\label{tab:QMC04}
\end{table}

\begin{table}[h!]
	\centering\footnotesize
	
	\begin{tabular}{ |c||c|c|c|c|c|c|c|c|c|c|} 
		\hline
		$p  \  \backslash  \  \zeta$ & $0.1$ & $0.2$ & $0.3$ & $ 0.4$  & $ 0.5$ & $0.6$ & $0.7$ & $0.8$ & $0.9$ & $1$ \\ 
		\hline
		
		$0.6$ &  &  & 1 &    &1  & 1 & 1 & 2 & 2 & 2 \\ 
		
		$0.7$ &  & 1 & 1 &  1  & 1 & 1 & 2 & 2 & 2 & 2 \\ 
		
		$0.8$ &  &  & 1 & 1   & 1 & 1 & 2 & 2 & 2 & 2 \\  
		
		$0.9$ &  & 1 & 1 & 1   & 1 & 1 & 2 & 2 & 2 & 2 \\  
		
		$1$ & 1 & 1 & 1 & 1  & 1 & 1 & 2 & 2 & 2 & 2 \\  
		
		\hline
	\end{tabular}
	
	\
	
	\centerline{Model Fitness: 1 for Exponential decay model and 2 for Rational decay model}
	
	\
	
	\begin{minipage}{.4\linewidth}
		\centering\footnotesize
		
		\begin{tabular}{ |c||c|c|c|c|c|} 
			\hline
			$\zeta  \  \backslash  \  p $& $0.6$ & $0.7$ & $0.8$ & $0.9$ & $1$ \\ 
			\hline
			
			$0.1$ &       &      &       &       & 0.23 \\
			
			$0.2$ &       &      &       & 0.25 & 0.19 \\
			
			$0.3$ & 0.64 &0.45 & 0.3 & 0.21 & 0.15 \\
			
			$0.4$ &       &0.38 & 0.24 & 0.16 & 0.11 \\
			
			$0.5$ & 0.49 &0.31 & 0.19 & 0.12 & 0.08 \\
			
			$0.6$ & 0.42 &0.24 & 0.14 & 0.08 & 0.05\\
			
			$0.7$ & 0.34 &  &  &  &  \\  
			
			$0.8$ & &  &  &  &  \\  
			
			$0.9$ &  &  &  &  &  \\  
			
			$1$ &  &  &  &  &  \\  
			
			\hline
		\end{tabular}
	
	\
	
	\centerline{Exponential Decay Rates}
	\end{minipage}\ \ \ \ \ \ \ \ \ \ \ \ \ \ \ \
	\begin{minipage}{.4\linewidth}
		\centering\footnotesize
		
		\begin{tabular}{ |c||c|c|c|c|c|} 
			\hline
			$\zeta  \  \backslash  \  p $& $0.6$ & $0.7$ & $0.8$ & $0.9$ & $1$ \\ 
			\hline
			
			$0.1$ &  &  &  &  &  \\
			
			$0.2$ &  &  &  &  &  \\
			
			$0.3$ &  &  &  &  &  \\
			
			$0.4$ &  & &  &  &  \\
			
			$0.5$ &  & &  &  &  \\
			
			$0.6$ &  & &  &  &  \\
			
			$0.7$ &  &0.85 & 0.71 & 0.61 & 0.52 \\  
			
			$0.8$ & 0.93 & 0.73 & 0.6 & 0.49 & 0.41 \\  
			
			$0.9$ & 0.79 & 0.6   &   0.48 & 0.39 & 0.32\\  
			
			$1$   & 0.64 & 0.47 & 0.37 & 0.29 & 0.23 \\  
			
			\hline
		\end{tabular}
		
		\
		
		\centerline{Rational Decay Rates}
	\end{minipage}

\caption{Estimated convergence rates $r$ for $\lambda=0.35$ by non-linear regression model}
\label{tab:QMC05}
\end{table}

\begin{table}[h!]
	\centering\footnotesize
	
	\begin{tabular}{ |c||c|c|c|c|c|c|c|c|c|c|} 
		\hline
		$p  \  \backslash  \  \zeta$ & $0.1$ & $0.2$ & $0.3$ & $ 0.4$  & $ 0.5$ & $0.6$ & $0.7$ & $0.8$ & $0.9$ & $1$ \\ 
		\hline
		
		$0.6$ & 1 & 1 & 1 & 1 &1  & 1 & 2 & 2 & 2 & 2 \\ 
		
		$0.7$ & 1 & 1 & 1 & 1 & 1 & 1 & 2 & 2 & 2 & 2 \\ 
		
		$0.8$ & 1 & 1 & 1 & 1 & 1 & 1 & 2 & 2 & 2 & 2 \\  
		
		$0.9$ & 1 & 1 & 1 & 1 & 1 & 1 & 2 & 2 & 2 & 2 \\  
		
		$1$ & 1 & 1 & 1 & 1  & 1 & 1 & 2 & 2 & 2 & 2 \\  
		
		\hline
	\end{tabular}
	
	\
	
	\centerline{Model Fitness: 1 for Exponential decay model and 2 for Rational decay model}
	
	\
	
	\begin{minipage}{.4\linewidth}
		\centering\footnotesize
		
		\begin{tabular}{ |c||c|c|c|c|c|} 
			\hline
			$\zeta  \  \backslash  \  p $& $0.6$ & $0.7$ & $0.8$ & $0.9$ & $1$ \\ 
			\hline
			 
			$0.1$ & 0.19  & 0.13 & 0.1   & 0.08 & 0.06 \\
			
			$0.2$ & 0.15  & 0.11  & 0.08  & 0.06 & 0.05 \\
			
			$0.3$ & 0.11  & 0.08  & 0.06  & 0.04 & 0.03 \\
			
			$0.4$ & 0.08  & 0.05 &  0.04  & 0.03 & 0.02 \\
			
			$0.5$ & 0.05  &0.03 & 0.02 & 0.02 & 0.01 \\
			
			$0.6$ & 0.03  &0.02 & 0.01 & 0.01 & 0.01\\
			
			$0.7$ &  &  &  &  &  \\  
			
			$0.8$ & &  &  &  &  \\  
			
			$0.9$ &  &  &  &  &  \\  
			
			$1$ &  &  &  &  &  \\  
			
			\hline
		\end{tabular}
		
		\
		
		\centerline{Rxponential Decay Rates}
	\end{minipage}\ \ \ \ \ \ \ \ \ \ \ \ \ \ \ \
	\begin{minipage}{.4\linewidth}
		\centering\footnotesize
		
		\begin{tabular}{ |c||c|c|c|c|c|} 
			\hline
			$\zeta  \  \backslash  \  p $& $0.6$ & $0.7$ & $0.8$ & $0.9$ & $1$ \\ 
			\hline
			
			$0.1$ &  &  &  &  &  \\
			
			$0.2$ &  &  &  &  &  \\
			
			$0.3$ &  &  &  &  &  \\
			
			$0.4$ &  & &  &  &  \\
			
			$0.5$ &  & &  &  &  \\
			
			$0.6$ &  & &  &  &  \\
			
			$0.7$ & 0.44 & 0.36 & 0.3 & 0.25 & 0.21 \\  
			
			$0.8$ & 0.34 & 0.27 & 0.22  & 0.18 & 0.15 \\  
			
			$0.9$ & 0.26 & 0.2  & 0.16 & 0.13 & 0.11\\  
			
			$1$   & 0.19 & 0.15 & 0.12 & 0.1 & 0.08 \\  
			
			\hline
		\end{tabular}
		
		\
		
		\centerline{Rational Decay Rates}
	\end{minipage}
	
\caption{Estimated convergence rates $r$ for $\lambda=0.2$ by non-linear regression model}
\label{tab:QMC06}
\end{table}

\subsection{High dimensional quantum Markov chain for dimension $N>2$}
It would be interesting to see if our observations above also hold for higher dimensional quantum Markov chain.
 
For the pure quantum case, we first observe from the simulation results that periodicity behavior remains and the probability will concentrates on the diagonal of the matrix $P_n$ varying from $\frac{1}{N}$ to $1$ where $N$ is the dimension of the space and $P_n$ is a $N \times N$ matrix, and the off-diagonal entries vary from $0$ to $\frac{1}{N}$ due to the symmetry property (See Figure \ref{fig:QMC07}). 

On the other hand, the periodicity for high dimensional Markov chains also follows the theoretical formula (\ref{eq:PureQuantumPeriod1}) for 2 dimensional case. Figure \ref{fig:QMC08} shows that the ratio between the periods simulated and the theoretical period from (\ref{eq:PureQuantumPeriod1}) is asymptotically $1$ for $\lambda = 1$, $\zeta=1$ and space dimension $N=3,4,5$.

\begin{figure}[h!]%
	\centering
	\subfigure[$N=3$]{%
		\label{fig:QMC07a}%
		\includegraphics[width=2.7in,height=2.7in]{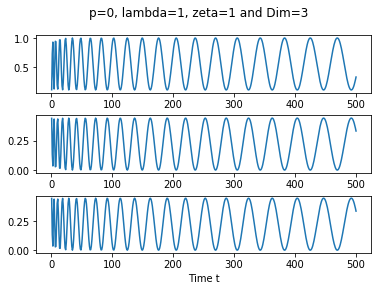}}%
	\hfill
	\subfigure[$N=5$]{%
		\label{fig:QMC07b}%
		\includegraphics[width=2.7in,height=2.7in]{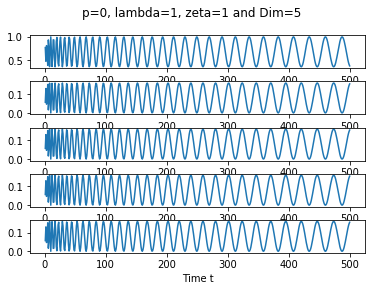}}%
	
	\caption{Simulated probabilities for $P_n(1,:)$ for $p=0$, $\lambda=1$, $\zeta=1$, dimension $N$ and $t=500$}
	\label{fig:QMC07}
\end{figure}

\begin{figure}[h!]%
	\centering
	\subfigure[$N=3$]{%
		\label{fig:QMC08a}%
		\includegraphics[width=1.9in,height=1.9in]{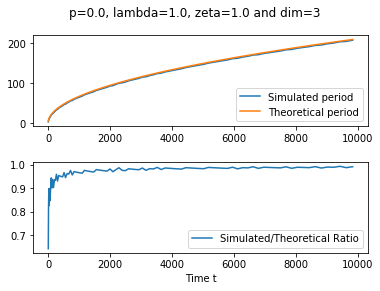}}%
	\hfill
	\subfigure[$N=4$]{%
		\label{fig:QMC08b}%
		\includegraphics[width=1.9in,height=1.9in]{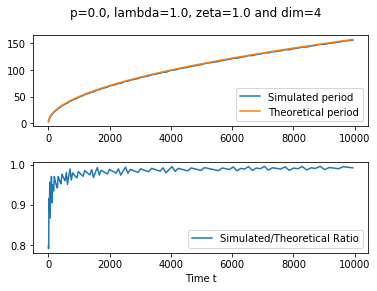}}%
	\hfill
	\subfigure[$N=5$]{%
		\label{fig:QMC08c}%
		\includegraphics[width=1.9in,height=1.9in]{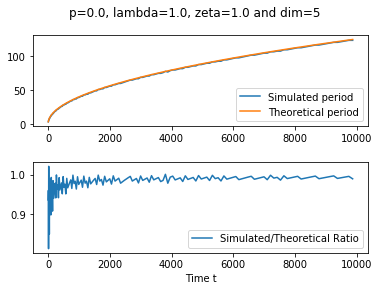}}%	
	\caption{Simulated vs pure quantum theoretical periods for $P_n$, $\lambda=1$ and $\zeta=1$, dimension $N$ with $t=2000$}
	\label{fig:QMC08}
\end{figure}

Consider now the decoherent case when $p$ near $0$. Unlike the symmetrical decay to the equilibrium of the $2$ dimensional case, the local maximums of the diagonal decrease to the equilibrium limit from $1$ to $\frac{1}{N}$ while the local minimums maintain near $\frac{1}{N}$. On the other hand, for the off-diagonal entries, the local minimums increase to the equilibrium limit from $0$ to $\frac{1}{N}$. Figure \ref{fig:QMC09} illustrates the $4$ dimensional decoherent Markov chain convergence to the equilibrium of the first row of the matrix $P_n$ when $p=0.005$, $\lambda=1$, $\zeta=0.5$ and $t=1000$.

\begin{figure}[h!]%
	\centering
	
	\includegraphics[width=2.5in,height=2.5in]{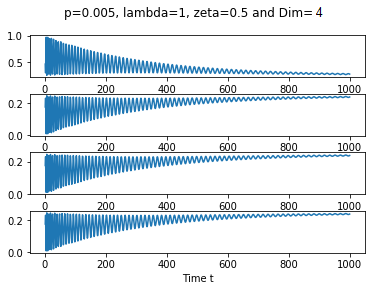}%
	
	\caption{Simulated probabilities for $P_n(1,:)$ for $p=0.005$, $\lambda=1$, $\zeta=0.5$, dimension $N=4$ and $t=1000$}
	\label{fig:QMC09}
\end{figure}

Similarly, we fit the local maximums of the probabilities at the diagonal entries of the matrix $P_n$ in non-linear regression with exponential decay model: $ce^{-rt}-\frac{1}{N}$ and rational decay model: $ct^{-rt}-\frac{1}{N}$ to estimate the decay rates, and compare their adjusted $R^2$ coefficients. As results, the exponential decay model has better fit than rational model with adjusted $R^2$ coefficients greater than $0.9$ for different values of $\zeta$ and $\lambda$ for $p\sim 0$. Even further, the convergence rates near $p\sim 0$ has the similar behavior as $2$ dimensional case, they are proportional to $p$ by the relation $r=\frac{p}{2}$, and independent of the other parameters (See Table \ref{tab:QMC07}).

\begin{table}[h!]
	\centering%\small
	\footnotesize
	\begin{tabular}{ |c||c|c|c|c|c|c|c|c|c|c|} 
		\hline
		$p  \  \backslash  \  \zeta$ & $0.1$ & $0.2$ & $0.3$ & $ 0.4$  & $ 0.5$ & $0.6$ & $0.7$ & $0.8$ & $0.9$ & $1$ \\ 
		\hline
		$0.005$ & 0.0024 & 0.0024 & 0.0025 & 0.0025   & 0.0025 & 0.0025 & 0.0025 & 0.0025 & 0.0025 & 0.0025  \\ 
		
		$0.01$ & 0.0047 & 0.0048 & 0.0049 & 0.005   & 0.005 & 0.0051 & 0.005 & 0.0051 & 0.0051 & 0.005 \\  
		
		$0.015$ & 0.0071 & 0.0072 & 0.0074 & 0.0075   & 0.0076 & 0.0075 & 0.0076 & 0.0075 & 0.0076 & 0.0075 \\  
		
		$0.02$ & 0.0095 & 0.0095 & 0.0099 & 0.0099   & 0.01 & 0.0102 & 0.0101 & 0.0101 & 0.0101 & 0.01 \\  
		
		$0.025$ & 0.0119 & 0.0119 & 0.0124 & 0.0125  & 0.0126 & 0.0128 & 0.0126 & 0.0127 & 0.0126 & 0.0125 \\  
		
		%$0.03$ & 0.0151 & 0.0153 & 0.015 & 0.0152   & 0.0152 & 0.015 & 0.015 &  0.0152  &  0.0153  & 0.0151 \\  
		\hline
	\end{tabular}
	\caption{Estimated $5$ dimensional convergence rates $r$ for $\lambda=0.5$  by non-linear regression model: $ce^{-rt}$}
	\label{tab:QMC07}
\end{table}

For $p$ close to $1$, we also fit the diagonal entries of the matrix $P_n$ in the same non-linear regression model, the same phase of transition point we obtained in the $2$ dimensional case also occurs in high dimensional case. Table \ref{tab:QMC08} shows the best fit non-linear regression results: exponential decay model or rational decay model with $R^2$ coefficients greater than $0.9$ and their estimated rates $r$ with respectively as before, and we numerically obtained that the phase of transition is between $0.6<\zeta<0.7$ (See Table \ref{tab:QMC08}).

\begin{table}[h!]
	\centering\footnotesize
	
	\begin{tabular}{ |c||c|c|c|c|c|c|c|c|c|c|} 
		\hline
		$p  \  \backslash  \  \zeta$ & $0.1$ & $0.2$ & $0.3$ & $ 0.4$  & $ 0.5$ & $0.6$ & $0.7$ & $0.8$ & $0.9$ & $1$ \\ 
		\hline
		
		$0.6$ & 1 & 1 & 1 & 1 &1  & 1 & 1 & 2 & 2 & 2 \\ 
		
		$0.7$ & 1 & 1 & 1 & 1 & 1 & 1 & 2 & 2 & 2 & 2 \\ 
		
		$0.8$ & 1 & 1 & 1 & 1 & 1 & 1 & 2 & 2 & 2 & 2 \\  
		
		$0.9$ & 1 & 1 & 1 & 1 & 1 & 1 & 2 & 2 & 2 & 2 \\  
		
		$1$ & 1 & 1 & 1 & 1  & 1 & 1 & 2 & 2 & 2 & 2 \\  
		
		\hline
	\end{tabular}
	
	\
	
	\centerline{Model Fitness: 1 for Exponential decay model and 2 for Rational decay model}
	
	\
	
	\begin{minipage}{.4\linewidth}
		\centering\footnotesize
		
		\begin{tabular}{ |c||c|c|c|c|c|} 
			\hline
			$\zeta  \  \backslash  \  p $& $0.6$ & $0.7$ & $0.8$ & $0.9$ & $1$ \\ 
			\hline
			
			$0.1$ & 0.29  & 0.27 & 0.24   & 0.2 & 0.17 \\
			
			$0.2$ & 0.26  & 0.24  & 0.2  & 0.17 & 0.14 \\
			
			$0.3$ & 0.23  & 0.2  & 0.17  & 0.14 & 0.11 \\
			
			$0.4$ & 0.2  & 0.17 &  0.14  & 0.1 & 0.08 \\
			
			$0.5$ & 0.17  &0.14 & 0.1 & 0.07 & 0.05 \\
			
			$0.6$ & 0.13  &0.1 & 0.07 & 0.05 & 0.03\\
			
			$0.7$ &  &  &  &  &  0.01\\  
			
			$0.8$ & &  &  &  &  \\  
			
			$0.9$ &  &  &  &  &  \\  
			
			$1$ &  &  &  &  &  \\  
			
			\hline
		\end{tabular}
		
		\
		
		\centerline{Exponential Decay Rates}
	\end{minipage}\ \ \ \ \ \ \ \ \ \ \ \ \ \ \ \
	\begin{minipage}{.4\linewidth}
		\centering\footnotesize
		
		\begin{tabular}{ |c||c|c|c|c|c|} 
			\hline
			$\zeta  \  \backslash  \  p $& $0.6$ & $0.7$ & $0.8$ & $0.9$ & $1$ \\ 
			\hline
			
			$0.1$ &  &  &  &  &  \\
			
			$0.2$ &  &  &  &  &  \\
			
			$0.3$ &  &  &  &  &  \\
			
			$0.4$ &  & &  &  &  \\
			
			$0.5$ &  & &  &  &  \\
			
			$0.6$ &  & &  &  &  \\
			
			$0.7$ & 0.4 & 0.47 & 0.55 & 0.63 &  \\  
			
			$0.8$ & 0.62 & 0.53 & 0.45  & 0.38 & 0.32 \\  
			
			$0.9$ & 0.52 & 0.43  & 0.36 & 0.3 & 0.25\\  
			
			$1$   & 0.43 & 0.35 & 0.28 & 0.23 & 0.19 \\  
			
			\hline
		\end{tabular}
		
		\
		
		\centerline{Rational Decay Rates}
	\end{minipage}
	
	\caption{Estimated $5$ dimensional convergence rates $r$ for $\lambda=0.2$ by non-linear regression model}
	\label{tab:QMC08}
\end{table}

\subsection{Decoherent case, $p>0$, critical behavior $\zeta_c=1$}
Recall our mainly result (Theorem \ref{finalEqui}) states that for $0<\zeta\leq 1$, the $N$ dimensional time-inhomogeneous decoherent Markov chains converge to theirs equilibrium limit, $\frac{1}{N}$. However, we expect that the critical transition occurs exactly when $\zeta_c=1$ which means that if $\zeta>1$, the large time scale limit is not $\frac{1}{N}$ and depends on the the values of $\lambda$ and the decoherence parameter $p$.

Figure \ref{fig:QMC11} support our expectation and shows the simulation results for $2$ dimensional $P_n(1,1)$ with $\zeta=1.1$ and different values of $\lambda$ and the decoherent parameter $p$. We observe from these figures that for fixed $\zeta$, the probabilities converge, and they don't converge to $\frac{1}{2}$ which is uniform distribution. In other words, the large time scale limits depend on the values of $\zeta$, $\lambda$ and the decoherence parameter $p$. More specifically, the large time scale limits depend on the initial conditions.

\begin{figure}[h!]%
	\centering
	\subfigure[$\lambda=0.3$]{%
		\label{fig:QMC11a}%
		\includegraphics[width=1.9in,height=1.9in]{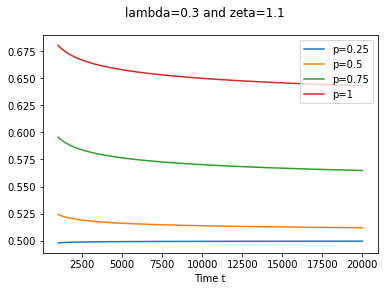}}%
	\hfill
	\subfigure[$\lambda=0.5$]{%
		\label{fig:QMC11b}%
		\includegraphics[width=1.9in,height=1.9in]{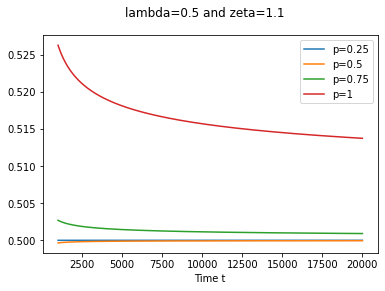}}%
	\hfill
	\subfigure[$\lambda=0.7$]{%
		\label{fig:QMC11c}%
		\includegraphics[width=1.9in,height=1.9in]{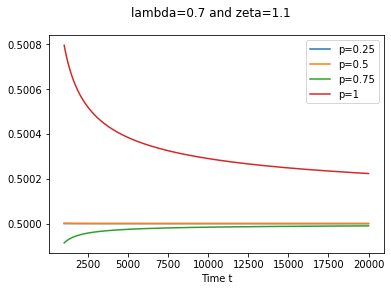}}%	
	\caption{Convergence of $2$ dimensional $P_n(1,1)$ fixing $\zeta=1.1$ for different $\lambda$'s and $p$'s}
	\label{fig:QMC11}
\end{figure} 

\subsection{Time-inhomogeneous quantum Markov chain with cyclic graph}

Instead of the matrix $G$ with strictly positive entries we considered previously, let us consider the quantum Markov chain with cyclic graphs. Consider the same decoherent quantum Markov chain model with the matrix 

$$G=
	\left[
	\begin{array}{ c c c c c}
		0 & \lambda & 0 & 0 & \lambda  \\
		\lambda & 0 & \lambda & 0 & 0  \\
		0 & \lambda & 0 & \lambda & 0  \\
		0 & 0 & \lambda & 0 & \lambda  \\
		\lambda & 0 & 0 & \lambda & 0  
	\end{array} \right]
$$
where $\lambda>0$.

For the pure quantum $p=0$ case, simulation results show that the time-inhomogeneous periodic behavior exists. Moreover, probabilities still concentrate more at the diagonal entries of $P_n$ varying near $\frac{1}{N}$ to 1 while the probabilities on the off-diagonal entries vary from $0$ to values near $\frac{1}{N}$ due to the symmetry. 

On the other hand, the probabilities at each entry also numerically converge to $\frac{1}{N}$ for relatively small $\zeta$'s and $\lambda$'s. 

Figure \ref{fig:QMC10a} illustrates the periodicity with local maximums and minimums of the first column of the matrix $P_n$ for $5$ dimensional pure cyclic quantum Markov chain with $\lambda=\frac{1}{2}$ and $\zeta=\frac{1}{2}$, Figures  
\ref{fig:QMC10b} and \ref{fig:QMC10c} show the convergence to equilibrium property for $5$ dimensional decoherent cyclic quantum Markov chain to $\frac{1}{5}$ with the same parameters when $p=0.01$ and the completely decoherent case $p=1$. Although G does not have all positive entries, we still believe that the transition points are $\zeta_c=1$ and $\zeta_o=2$ for decoherent case and pure quantum case respectively by similar argument in section \ref{section:timInhomoPeriod}.

Furthermore, we believe that the phase of transition occurs at $\zeta_c=1$ and $\zeta_o=2$ for any irreducible $G$.

\begin{figure}[h!]%
	\centering
	\subfigure[$p=0$, $\lambda=0.5$, $\zeta=0.5$]{%
		\label{fig:QMC10a}%
		\includegraphics[width=1.9in,height=1.9in]{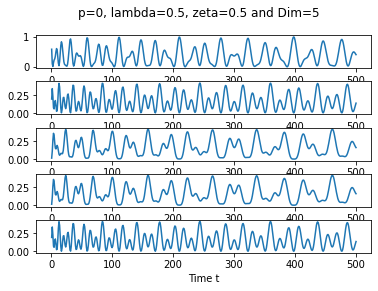}}%
	\hfill
	\subfigure[$p=0.01$, $\lambda=0.5$, $\zeta=0.5$]{%
		\label{fig:QMC10b}%
		\includegraphics[width=1.9in,height=1.9in]{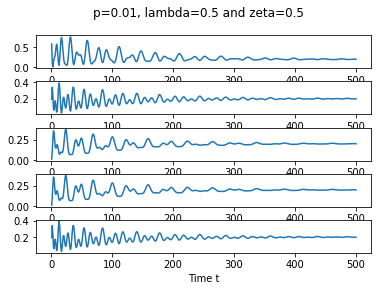}}%
	\hfill
	\subfigure[$p=1$, $\lambda=0.5$, $\zeta=0.5$]{%
		\label{fig:QMC10c}%
		\includegraphics[width=1.9in,height=1.9in]{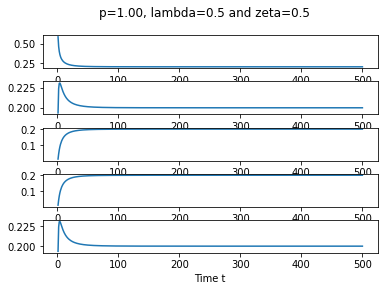}}%
	\caption{Simulated probabilities for $P_n(1,:)$ using cyclic graph $G$ with dimension $5$ and $t=500$}
	\label{fig:QMC10}
\end{figure}

\section{Conclusion and future work}\label{sec:conclusion}
In this paper, we mainly considered the time-inhomogeneous unitary operators, and defined the time-inhomogeneous quantum analogue of the classical Markov chain with decoherence parameter on finite dimensional state space and we interpreted the decoherent parameter as the probability to perform a measurement, that means that at each step, we perform a measurement with a certain probability. We proved the Markov properties at the geometric measurement times using the path integral representation.

We also made the conclusion that the time-inhomogeneous quantum Markov chain on finite dimensional state spaces with non zero probability of measurement converges to an equilibrium limit when $0\leq \zeta \leq 1$ as time approaches to infinity.

Additionally, we analyzed the periodicity, the decay rate and the convergence of our quantum Markov chains. We numerically concluded that our quantum Markov chain exhibits a time-inhomogeneous periodic and non-periodic behaviors with phase of transition $\zeta_o=2$ for the probability of measurement close to zero, and a large time scale equilibrium behavior with phase of transition $\zeta_c=1$ for the positive probability of measurement. Moreover, the decoherent Markov chain decays exponentially for small $\zeta$ and decays with power rule for large $\zeta$ with critical point $\zeta_d$ between $0.6$ and $0.7$.

Instead of fully connected graph, numerical simulations were also done considering the cyclic graph, and illustrated similar periodic behavior for coherent quantum Markov chain and convergence to equilibrium for decoherent quantum Markov chain for small $\lambda$'s $\zeta$'s. As result, we expect that the phase of transition occurs at $\zeta_c=1$ and $\zeta_o=2$ for cyclic graphs.

Not to mention, our work is not only based on finite quantum state space, also called "Qudit" in high dimensional quantum computation which researchers have developed many applications recently including quantum circuit building, quantum algorithm design and quantum experimental methods (see \cite{WangQudits}), but might also be related to one of the fundamental aspects of quantum annealing algorithm which was first proposed by B. Apolloni, N. Cesa Bianchi and D. De Falco in \cite{ApolloniQuantumAnnealing} and \cite{ApolloniQuantumStochastic} and later formulated by T. Kadowaki and H. Nishimori in \cite{KadowakiQuantumAnnealing}. With this motivation, our future work will focus on applications of quantum search algorithms with time-inhomogeneous quantum Markov chain in high dimensional spaces or general graph. 

For $0<p\leq 1$ and $\zeta=0$, our model is related to open quantum random walks in \cite{AttalOpenQuantum1} and \cite{AttalOpenQuantum2} which has a lot of potential applications in quantum computing. It would be very interesting to see how open quantum walk can be generalized to time-inhomogeneous case.

Furthermore, our numerical results will also lead us challenging analytical problems, e.g., generalizing our results to time-inhomogeneous quantum Markov chain associated with any connected graphs.

\end{document}